\documentclass[journal]{IEEEtran}

\usepackage[]{cite}
    
\usepackage{tikz}
\usetikzlibrary{calc,positioning}
\usepackage{amsmath}
\usepackage{amsthm}
\usetikzlibrary{patterns}
\usepackage{graphicx}
\usepackage{ifthen}
\usepackage{tikz-3dplot}
\usepackage{esint}

\usepackage{soul}

\usepackage{amssymb}

\newcommand{\storage}[1]{\E{\left(#1\right)}}
\newcommand{\supply}[1]{s{\left(#1\right)}}

\newcommand{\vect}[1]{\mathbf{#1}}

\newcommand{\matr}[1]{\mathbf{#1}}

\newcommand{\junk}[1] {}

\def\XXint#1#2#3{{\setbox0=\hbox{$#1{#2#3}{\int}$}
\vcenter{\hbox{$#2#3$}}\kern-.5\wd0}}

\newcommand*\widebar[1]{%
  \hbox{%
    \vbox{%
      \hrule height 0.5pt 
      \kern0.3ex
      \hbox{%
        \kern-0.05em
        \ensuremath{#1}%
        \kern-0.05em
      }%
    }%
  }%
}

\newcommand{\half}{\frac{1}{2}}
\newcommand{\halftxt}{(1/2)}

\newcommand{\ihat}{\hat{\imath}}
\newcommand{\jhat}{\hat{\jmath}}
\newcommand{\khat}{\hat{k}}

\newcommand{\I}{\matr{I}}

\newcommand{\x}{\vect{x}}
\newcommand{\R}{\matr{R}}
\newcommand{\F}{\matr{F}}
\newcommand{\Lm}{\matr{L}}
\newcommand{\Ehat}{\hat{E}}
\newcommand{\Evhat}{\hat{\Ev}}
\newcommand{\Uhat}{\hat{U}}
\newcommand{\Uvhat}{\hat{\Uv}}

\newcommand{\Hvhat}{\hat{\Hv}}

\newcommand{\E}{\mathcal{E}}

\newcommand{\dx}{\Delta x}
\newcommand{\dy}{\Delta y}
\newcommand{\dz}{\Delta z}

\newcommand{\dxhalf}{\frac{\dx}{2}}
\newcommand{\dyhalf}{\frac{\dy}{2}}
\newcommand{\dzhalf}{\frac{\dz}{2}}

\newcommand{\Ev}{\vect{E}}
\newcommand{\Hv}{\vect{H}}
\newcommand{\Uv}{\vect{U}}

\newcommand{\eps}{\varepsilon}
\newcommand{\Q}{\matr{Q}}

\pgfmathsetmacro{\fieldarrthick}{0.6}
\newcommand{\mhat}{\hat{m}}
\definecolor{HV}{rgb}{0,0.4,0}

\definecolor{Hcol}{rgb}{0.7,0,0}
\definecolor{EB}{rgb}{0.7,0,0}
\definecolor{Ecol}{rgb}{0,0,1}

\definecolor{Secondary}{rgb}{0.55,0.55,0.55}
\definecolor{Shaded}{rgb}{0.8,0.8,0.8}

\definecolor{legendGreen}{rgb}{0,0.5,0}
\definecolor{legendRed}{rgb}{1,0,0}
\definecolor{legendBlue}{rgb}{0,0,1}
\definecolor{legendGrey}{rgb}{0.5,0.5,0.5}

\usepackage [english]{babel}

\usepackage[para,online,flushleft]{threeparttable}

\newcommand{\Sign}{\matr{P}}
\newcommand{\Singvals}{\matr{\Sigma}}
\newcommand{\IO}{\matr{S}}

\newcommand{\T}{\matr{T}}
\newcommand{\one}{\matr{1}}
\newcommand{\rx}{{r_x}}
\newcommand{\ry}{{r_y}}

\newcommand{\dt}{\Delta t}

\newcommand{\D}{\matr{D}}

\newcommand{\DDH}{\D}
\newcommand{\DDU}{\D}
\newcommand{\DDE}{\D}
\newcommand{\lE}{l}
\newcommand{\lU}{l'}
\newcommand{\lH}{l'}

\newcommand{\Dlx}{\DDE_{\lE_{x}}}
\newcommand{\Dly}{\DDE_{\lE_{y}}}
\newcommand{\Dlz}{\DDE_{\lE_{z}}}
\newcommand{\Dl}{\DDE_{\l}}

\newcommand{\DlpU}{\DDU_{l'_U}}
\newcommand{\Dlp}{\DDH_{l'}}

\newcommand{\Dlpx}{\DDH_{\lH_{x}}}
\newcommand{\Dlpy}{\DDH_{\lH_{y}}}
\newcommand{\Dlpz}{\DDH_{\lH_{z}}}

\newcommand{\DlpyB}{\DDU_{\lU_{y} \Bottom}}
\newcommand{\DlpyT}{\DDU_{\lU_{y} \Top}}
\newcommand{\DlpzS}{\DDU_{\lU_{z} \South}}
\newcommand{\DlpzN}{\DDU_{\lU_{z} \North}}

\newcommand{\aE}{A'}

\newcommand{\aH}{A}

\newcommand{\DAp}{\D_{A'}}

\newcommand{\DAx}{\D_{\aH_{x}}}
\newcommand{\DAy}{\D_{\aH_{y}}}
\newcommand{\DAz}{\D_{\aH_{z}}}
\newcommand{\DA}{\D_{A}}

\newcommand{\DApx}{\D_{\aE_{x}}}

\newcommand{\Dex}{\D_{\varepsilon_x}}

\newcommand{\Dhat}{\hat{\D}}

\newcommand{\Dmx}{\D_{\mu_x}}
\newcommand{\Dmy}{\D_{\mu_y}}
\newcommand{\Dmz}{\D_{\mu_z}}
\newcommand{\Dm}{\D_{\mu}}

\newcommand{\Dsx}{\D_{\sigma_x}}

\newcommand{\BExB}{\matr{B}_{{x\Bottom}}}
\newcommand{\BExT}{\matr{B}_{{x \Top}}}
\newcommand{\BExS}{\matr{B}_{{x \South}}}
\newcommand{\BExN}{\matr{B}_{{x \North}}}

\newcommand{\BEyB}{\matr{B}_{{y \Bottom}}}
\newcommand{\BEyT}{\matr{B}_{{y \Top}}}
\newcommand{\BEyW}{\matr{B}_{{y \West}}}
\newcommand{\BEyE}{\matr{B}_{{y \East}}}

\newcommand{\BEzS}{\matr{B}_{{z \South}}}
\newcommand{\BEzN}{\matr{B}_{{z \North}}}
\newcommand{\BEzW}{\matr{B}_{{z \West}}}
\newcommand{\BEzE}{\matr{B}_{{z \East}}}

\newcommand{\GxEy}{\matr{G}_{xy}}
\newcommand{\GxEz}{\matr{G}_{xz}}

\newcommand{\GyEx}{\matr{G}_{yx}}
\newcommand{\GyEz}{\matr{G}_{yz}}

\newcommand{\GzEx}{\matr{G}_{zx}}
\newcommand{\GzEy}{\matr{G}_{zy}}

\newcommand{\Bottom}{,B}
\newcommand{\Top}{,T}
\newcommand{\North}{,N}
\newcommand{\East}{,E}
\newcommand{\South}{,S}
\newcommand{\West}{,W}

\newcommand{\HyB}{\Uv_{y \Bottom}}
\newcommand{\HyT}{\Uv_{y \Top}}
\newcommand{\HzS}{\Uv_{z \South}}
\newcommand{\HzN}{\Uv_{z \North}}

\newtheorem{theorem}{Theorem}

\newtheorem{definition}{Definition}

\newcommand{\phantomplus}{\phantom{+}}

\begin{document}

\title{A Dissipation Theory for Three-Dimensional\\FDTD with Application to Stability\\Analysis and Subgridding}

\author{Fadime~Bekmambetova,~\IEEEmembership{Student Member,~IEEE}, 
        Xinyue~Zhang~\IEEEmembership{Student Member,~IEEE},
        and~Piero~Triverio,~\IEEEmembership{Senior Member,~IEEE}
\thanks{Manuscript received ...; revised ...}%
\thanks{This work was supported in part by the Natural Sciences and Engineering Research Council of
Canada (Discovery grant program) and in part by the Canada Research Chairs program.}
\thanks{F.~Bekmambetova, X.~Zhang and P.~Triverio are with the Edward S. Rogers Sr. Department of Electrical and Computer Engineering, University of Toronto, Toronto, M5S~3G4 Canada (email: fadime.bekmambetova@mail.utoronto.ca, xinyuezhang.zhang@mail.utoronto.ca, piero.triverio@utoronto.ca).}
}

\markboth{IEEE Transactions on Antennas and Propagation}%
{Bekmambetova, Zhang, Triverio}

\maketitle

\IEEEpeerreviewmaketitle

\begin{abstract}
The finite-difference time-domain (FDTD) algorithm is a popular numerical method for solving electromagnetic problems. FDTD simulations can suffer from instability due to the explicit nature of the method. Stability enforcement can be particularly challenging in scenarios where a setup is composed of multiple components, such as  grids of different resolution, advanced boundary conditions,  reduced-order models, and lumped elements. We propose a dissipation theory for 3-D FDTD inspired by the principle of energy conservation. We view the FDTD update equations for a 3-D region as a dynamical system, and show under which conditions the system is dissipative. By requiring each component of an FDTD-like scheme to be dissipative, the stability of the overall coupled scheme follows by construction. The proposed framework enables the creation of provably stable schemes in an easy and modular fashion, since conditions are imposed on the individual components, rather than on the overall coupled scheme as in existing approaches. With the proposed framework, we derive a new subgridding scheme with guaranteed stability, low reflections, support for material traverse and arbitrary (odd) grid refinement ratio.

\end{abstract}

\def\figurename{Fig.}

\section{Introduction}
\label{sec:intro}

Stability is a major challenge in explicit numerical schemes for solving differential equations. The conventional finite-difference time-domain (FDTD) algorithm for solving Maxwell's equations~\cite{Yee} is stable when the iteration time step is below the Courant-Friedrichs-Lewy (CFL) limit~\cite{Gedney}. However, stability enforcement becomes much more difficult in the case of more advanced FDTD schemes such as, for example, those including locally-refined grids~\cite{okoniewski1997three,thoma1996consistent,xiao2007three}, reduced-order models~\cite{kulas2001rom,denecker}, and hybridizations of FDTD with integral equation methods~\cite{bretones1998hybrid} or circuit simulators~\cite{Erdin-2000-FDTD-MOR-SPICE}. Many of these advanced schemes can be seen as the coupling of multiple blocks, such as FDTD grids of different resolution, reduced-order models, lumped components, boundary conditions, and so on. While the stability of each individual block may be well understood, analyzing and enforcing the stability of the overall coupled scheme can be a daunting task.

Several techniques have been proposed for FDTD stability analysis, such as von Neumann analysis~\cite{taflove2005computational}, the iteration method~\cite{Gedney} and the energy method~\cite{edelvik2004general}. Unfortunately, von Neumann analysis cannot be applied in presence of inhomogeneous materials. The iteration method~\cite{Gedney, Remis} analyzes the eigenvalues of a matrix describing FDTD iterations. Even though an FDTD setup can often be divided into several subsystems, such as grids of different resolution, reduced-order models, interpolation conditions, and other blocks, the iteration matrix needs to be derived for the entire scheme, which can lead to lengthy derivations and make stability analysis very tedious.  A similar issue arises in the case of the energy method~\cite{edelvik2004general}, where one must write an expression for the stored energy in the entire domain, and show that the coupled scheme satisfies the principle of energy conservation. 

In this paper, we present a dissipation theory for modular stability enforcement of complex 3-D FDTD systems, generalizing previous work in two dimensions~\cite{jnl-2017-fdtd-dissipative}. The method is based on the theory of dissipative dynamical systems~\cite{willems1972dissipative}. First, we write the FDTD update equations for an arbitrary inhomogeneous region in the form of a discrete-time dynamical system. The tangential magnetic and electric field are seen as input and output, respectively. Suitable discrete expressions for the energy stored inside the region, and for the energy absorbed through the boundaries are proposed and used to investigate under which conditions the system is dissipative. This result becomes the basis for a powerful framework to create new FDTD algorithms with guaranteed stability. By imposing each subsystem  to to be dissipative, the overall coupled scheme formed by these subsystems is dissipative by construction, and thus stable. The main virtue of this approach is that the stability of the overall scheme follows automatically from conditions imposed on each subsystem \emph{individually}. This makes stability analysis simpler and more modular, which is a remarkable advantage over existing approaches that require the analysis of the overall coupled scheme. 
This modularity also facilitates the generation of new schemes, since proving a given subsystem (e.g. a sophisticated boundary condition) to be dissipative allows coupling it to any other dissipative block, with no need for additional analysis. Finally, we use the theory to derive a stable subgridding scheme for 3-D FDTD that works for any odd refinement ratio and naturally supports traversal of material boundaries by the subgridding interface.

This paper is structured as follows. In Sec.~\ref{sec:dyn_sys}, we show how the update equations for an FDTD region can be cast into the form of a dynamical system with suitable inputs and outputs. In Sec.~\ref{sec:fdtd_3d_dissipativity}, we provide dissipativity conditions for the system and show their relation to the CFL limit. In Sec.~\ref{sec:method}, we describe the modular method for enforcing FDTD stability using the proposed theory, and in Sec.~\ref{sec:subgridding} we apply this method to derive a stable subgridding scheme. Numerical examples are given in Sec.~\ref{sec:examples}.

\section{Dynamical System Formulation of a 3-D FDTD Region}
\label{sec:dyn_sys}

\begin{figure}[t]
	\centering
	
	\includegraphics[width=\columnwidth]{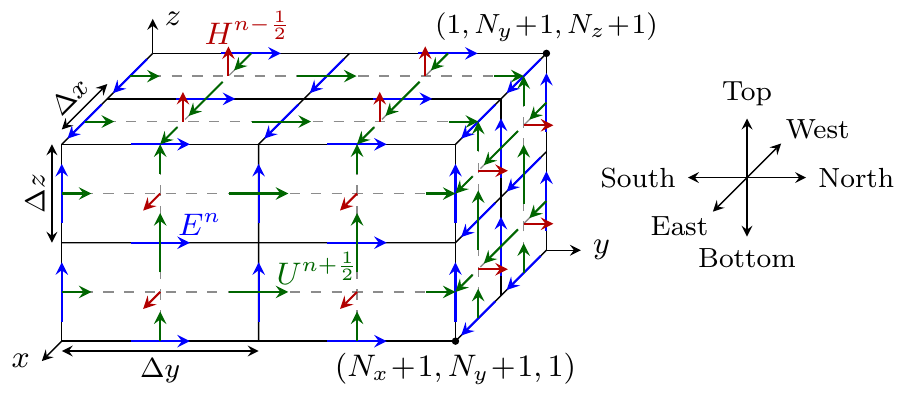}
	
	\caption{Illustration of the 3-D FDTD region with the hanging variables shown in green and the regular electric and magnetic samples in blue and red, respectively. The secondary grid is shown with the dashed line.}\label{fig:region}
\end{figure}

Consider the FDTD region shown in Fig.~\ref{fig:region}, which contains $N_x$, $N_y$, and $N_z$ primary cells in the $x$, $y$, and $z$ directions, respectively. Electric field is sampled on the primary grid edges at the integer time points, $n$, and magnetic field is sampled on the primary grid faces at the time points shifted by half of a time step, $n\!-\!\halftxt$~\cite{Gedney}. For simplicity, we assume that the region is discretized uniformly with primary cells of dimensions $\dx\!\times\!\dy\!\times\!\dz$. As shown in Fig.~\ref{fig:region}, we use cardinal directions in the $xy$-plane, with $+y$ being the North. ``Top'' and ``Bottom'' denote the $+z$ and $-z$ directions, respectively.

In addition to the regular electric and magnetic field samples, denoted by $E^n$ and $H^{n-\halftxt}$, respectively, we define the so-called hanging variables~\cite{venkatarayalu2007stable}, $U^{n+\halftxt}$, which are magnetic field samples collocated with the electric fields on the boundary of the region. The hanging variables, which are tangential to the boundary and perpendicular to the corresponding electric field samples, are used to define the power supplied to the region and to facilitate the derivation of a self-contained dynamical model for the region based on FDTD update equations.

In this section, we take the FDTD equations for the regular field samples, $E^n$ and $H^{n-\halftxt}$, and cast them in the form of a dynamical system with the hanging variables, $U^{n+\halftxt}$, as its inputs. The output consists of electric field samples on the boundary of the region.

\subsection{Equations at Each Node}
\label{subsec:eq_each_node}

\begin{figure*}[t]
	
	\centering
	\begin{tabular}{lll}
		(a)& \phantom{a}(b)& \phantom{a}(c)\\[5pt]
		\phantom{a}\includegraphics[scale=0.95]{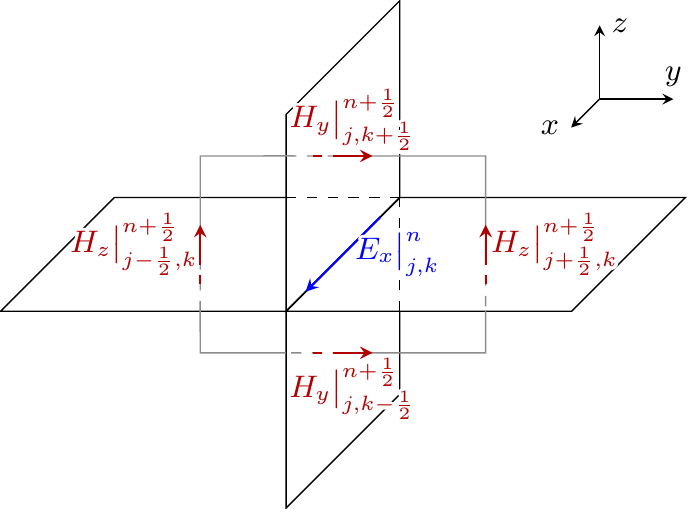}
		&
		\phantom{ab}\includegraphics[scale=0.95]{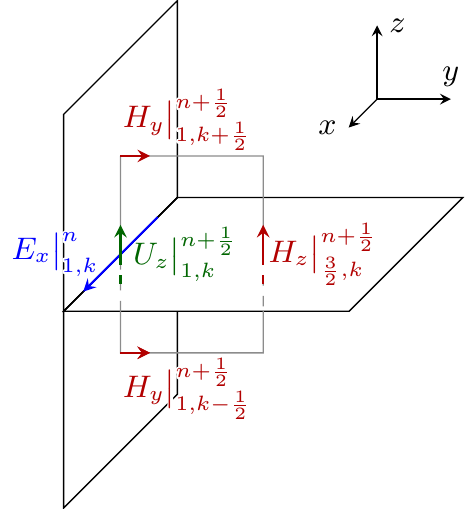}
		&
		\phantom{ac}\includegraphics[scale=0.95]{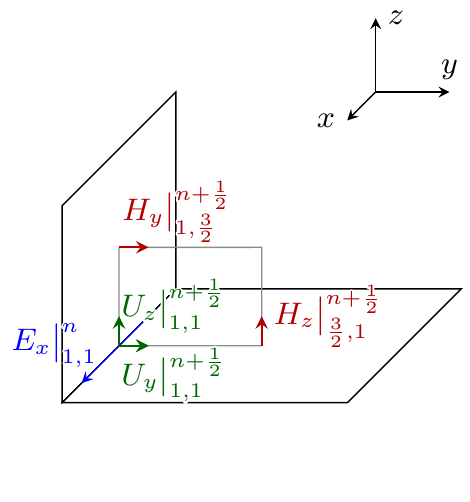}
	\end{tabular}
	\caption{Variables involved in the recurrence relations for (a) internal $E_x^n$ samples (Type~1), (b) $E_x^n$ samples of Type~2 on the South face, and (c) the $E_x^n$ samples of Type~3 on the Bottom-South edge of the boundary. The common subscript ``$i\!+\!\halftxt$'' is omitted for clarity.}
	\label{fig:update}
\end{figure*}

\begin{figure}[t!]
	\centering
	
	\hspace*{-1cm}\includegraphics[scale=1]{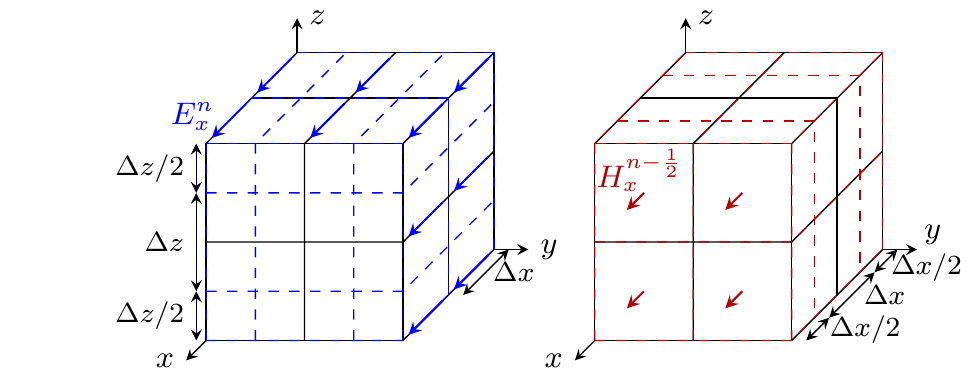}
	
	\caption{Illustration of the volumes (shown in dashed lines), associated with each electric (left panel) and magnetic (right panel) field sample in~\eqref{eq:update_Ex_internal}--\eqref{eq:update_Hx_face}.}\label{fig:volume_alloc}
\end{figure}

The temporal evolution of any electric field sample strictly inside the region, such as the $E_x|_{i+\halftxt,j,k}^n$ sample shown in Fig.~\ref{fig:update}a, is described using a conventional FDTD update equation~\cite{Gedney}
\begin{multline}
\dx\dy\dz
\left(\frac{\eps_x}{\dt}\!+\!\frac{\sigma_x}{2} \right) E_x\big|_{i+\half,j,k}^{n+1}\\
 = 
 \dx\dy\dz
 \left(\frac{\eps_x}{\dt}\!-\!\frac{\sigma_x}{2} \right)
 E_x\big|_{i+\half,j,k}^{n}\\
 + \dx\dz H_z\big|_{i+\half,j+\half,k}^{n+\half}
- \dx\dy H_y\big|_{i+\half,j,k+\half}^{n+\half} \\
-\dx\dz H_z\big|_{i+\half,j-\half,k}^{n+\half}
+\dx\dy H_y\big|_{i+\half,j,k-\half}^{n+\half}
\,,
\label{eq:update_Ex_internal}
\end{multline}
where $\sigma_x$ and $\eps_x$ are electrical conductivity and permittivity values, respectively, on the primary edge where the electric field is sampled. For clarity, we do not show the dependence of material properties on the location, although the proposed developments are valid in the most general case where all material properties are inhomogeneous. The iteration time step is denoted by $\dt$. Although $\dx$, the length of the edge where the sample is located, can be canceled in~\eqref{eq:update_Ex_internal}, we keep it for later derivations. The product $\dx\dy\dz$ is the volume associated with $E_x|_{i+\halftxt,j,k}^n$, as illustrated in Fig.~\ref{fig:volume_alloc}. We refer to the internal electric fields such as $E_x|_{i+\halftxt,j,k}^n$ as electric fields of Type~1.

A conventional FDTD equation for the electric fields on the region's faces, such as the South $E_x|_{i+\halftxt,1,k}^{n}$ sample in Fig.~\ref{fig:update}b, would involve some magnetic field samples that are beyond the considered region. The use of the magnetic fields outside the region would require assumptions on the nature of the subsystems connected to the FDTD grid. Instead, we write a modified equation for the South electric field samples using the hanging variables on the South boundary
\begin{multline}
\dx\dyhalf\dz\left(\frac{\eps_x}{\dt}\!+\!\frac{\sigma_x}{2} \right)E_x\big|_{i+\half,1,k}^{n+1} \\
= \dx\dyhalf\dz\left(\frac{\eps_x}{\dt}\!-\!\frac{\sigma_x}{2} \right)E_x\big|_{i+\half,1,k}^{n}\\
+ \dx\dz H_z\big|_{i+\half,\frac{3}{2},k}^{n+\half}
- \dx\dyhalf H_y\big|_{i+\half,1,k+\half}^{n+\half}\\
-\dx\dz U_z\big|_{i+\half,1,k}^{n+\half}
+\dx\dyhalf H_y\big|_{i+\half,1,k-\half}^{n+\half} \,.
\label{eq:update_Ex_south}
\end{multline}
Equation~\eqref{eq:update_Ex_south} is obtained from an FDTD-like approximation of Maxwell-Amp\`ere law on a half-cell containing $E_x|_{i+\halftxt,1,k}^{n}$. Using~\eqref{eq:update_Ex_south} instead of a conventional FDTD equation, we obtain a self-contained FDTD-like model for the fields in the region that does not involve field samples beyond its boundaries. This feature is crucial for investigating under which conditions the region is dissipative. It is also needed to obtain an FDTD model that can be connected to other subsystems, such as a grid of different resolution or a reduced model, where some magnetic field samples may not be available for a conventional FDTD equation. Samples like $E_x|_{i+\halftxt,1,k}^{n}$ are classified as Type~2, which are the electric field samples on the faces of the boundary. 

In a similar way, we write a modified FDTD equation for the samples located at the edges where two faces of the boundary come together. Those equations involve two hanging variables for each sample. For instance, in the case of the Bottom-South sample shown in Fig.~\ref{fig:update}c, the recurrence relation reads
\begin{multline}
\dx\dyhalf\dzhalf\left(\frac{\eps_x}{\dt}\!+\!\frac{\sigma_x}{2} \right)E_x\big|_{i+\half,1,1}^{n+1} \\
= \dx\dyhalf\dzhalf\left(\frac{\eps_x}{\dt}\!-\!\frac{\sigma_x}{2} \right)E_x\big|_{i+\half,1,1}^{n}\\
+ \dx\dzhalf H_z\big|_{i+\half,\frac{3}{2},1}^{n+\half}
- \dx\dyhalf H_y\big|_{i+\half,1,\frac{3}{2}}^{n+\half}\\
-\dx\dzhalf U_z\big|_{i+\half,1,1}^{n+\half} 
+\dx\dyhalf U_y\big|_{i+\half,1,1}^{n+\half}
\,.
\label{eq:update_Ex_edge}
\end{multline}
In~\eqref{eq:update_Ex_edge}, the $y$- and $z$-directed hanging variables, $U_y|_{i+\halftxt,1,1}^{n+\halftxt}$ and $U_z|_{i+\halftxt,1,1}^{n+\halftxt}$, respectively, are used to complete the line integral of magnetic field around $E_x|_{i+\halftxt,1,1}^{n}$. The sample $E_x|_{i+\halftxt,1,1}^{n}$ is an example of a Type~3 electric field, which is a field shared between two faces of the region's boundary. The three types of $E_y^n$ and $E_z^n$ fields and their recurrence relations are defined analogously to $E_x^n$ and to~\eqref{eq:update_Ex_internal}--\eqref{eq:update_Ex_edge}.

The recurrence relation for the magnetic $H_x^{n-\halftxt}$ samples that are strictly inside the region is the conventional update FDTD equation, which reads
\begin{multline}
\dx\dy\dz \frac{\mu_x}{\dt} H_x\big|_{i,j+\half,k+\half}^{n+\half} 
= \dx\dy\dz \frac{\mu_x}{\dt}H_x\big|_{i,j+\half,k+\half}^{n-\half}\\
-\dx\dz E_z\big|_{i,j+1,k+\half}^{n}
+\dx\dy E_y\big|_{i,j+\half,k+1}^{n} \\
+\dx\dz E_z\big|_{i,j,k+\half}^{n}
 -\dx\dy E_y\big|_{i,j+\half,k}^{n}
\,,
\label{eq:update_Hx_internal}
\end{multline}
where $\mu_x$ is the magnetic permeability at the sampling location. Equation for the boundary magnetic field sample $H_x|_{1,j+\halftxt,k+\halftxt}^{n-\halftxt}$, which is located at a secondary edge of length $\dx/2$ normal to the West boundary, involves a factor of $\dx/2$, as opposed to $\dx$
\begin{multline}
\dxhalf\dy\dz \frac{\mu_x}{\dt} H_x\big|_{1,j+\half,k+\half}^{n+\half} 
= \dxhalf\dy\dz \frac{\mu_x}{\dt}H_x\big|_{1,j+\half,k+\half}^{n-\half}\\
-\dxhalf\dz E_z\big|_{1,j+1,k+\half}^{n}
+\dxhalf\dy E_y\big|_{1,j+\half,k+1}^{n} \\
+\dxhalf\dz E_z\big|_{1,j,k+\half}^{n}
-\dxhalf\dy E_y\big|_{1,j+\half,k}^{n} \,,
\label{eq:update_Hx_face}
\end{multline}
and similarly for the boundary faces on the East side. Recurrence relations for $H_y^{n-\halftxt}$ and $H_z^{n-\halftxt}$ are obtained in a similar fashion.

The coefficients $\dx\dy\dz$ in~\eqref{eq:update_Hx_internal} and $(\dx/2)\dy\dz$ in~\eqref{eq:update_Hx_face} correspond to the dimensions of the volume associated with each $H_x^{n-\halftxt}$ sample, as shown in Fig.~\ref{fig:volume_alloc}.

\subsection{Compact Matrix Form}
\label{sec:mtx_form}
In order to shorten the notation, we write the equations~\eqref{eq:update_Ex_internal}--\eqref{eq:update_Hx_face} for all field samples in a matrix form. Collecting all recurrence relations for $H_x^{n-\halftxt}$ samples~\eqref{eq:update_Hx_internal} and~\eqref{eq:update_Hx_face} into matrix equations, we obtain
\begin{multline}
\Dlpx \DAx \frac{\Dmx}{\dt} \Hv_x^{n+\half} 
= \Dlpx \DAx \frac{\Dmx}{\dt} \Hv_x^{n-\half} \\
+\Dlpx \GzEy \Dly \Ev_y^n
-\Dlpx \GyEz \Dlz \Ev_z^n\,,
\label{eq:mtxHx}
\end{multline}
where vectors $\Hv_x^{n-\halftxt}$ and $\Ev_y^n$ contain all $H_x^{n-\halftxt}$ and $E_y^n$ samples in the region, respectively. The coefficient matrix $\Dmx$ is a diagonal matrix containing the magnetic permeability values on the edges where the corresponding elements of $\Hv_x^{n-\halftxt}$ are sampled. Matrix $\GzEy$, which consists of zeros, $+1$'s and $-1$'s, is a discrete $z$-derivative operator for $\Ev_y^n$. Similarly, $\GyEz$ is a discrete $y$-derivative operator for $\Ev_z^n$. Matrix $\Dlpx$ is a diagonal matrix containing the length of each secondary edge associated with samples in $\Hv_x^{n-\halftxt}$, namely $\dx$ for internal samples and $\dx/2$ for the samples on the East and West boundaries. Diagonal matrix $\DAx$ contains the area of primary faces where $H_x^{n-\halftxt}$ is sampled, which reduces to $\dy \dz \I_{N_{H_x}}$ for a uniform discretization, where $\I_m$ denotes an $m\times m$ identity matrix and $N_{H_x}$ is the number of $H_x^{n-\halftxt}$ samples in the region. As a result, the product $\Dlpx \DAx$ contains the volumes shown in the right panel of Fig.~\ref{fig:volume_alloc}, associated with each sample in $\Hv_x^{n-\halftxt}$. Matrices $\Dly$ and $\Dlz$ contain the length of the primary edges where the elements of $\Ev_y^n$ and $\Ev_z^n$ are sampled, respectively. These matrices reduce to $\dy \I_{N_{E_y}}$ and $\dz \I_{N_{E_z}}$ for uniform discretization.

In a similar fashion, we write in matrix form the recurrence relations for $\Hv_y^{n-\halftxt}$ and $\Hv_z^{n-\halftxt}$ as
\begin{multline}
\Dlpy \DAy \frac{\Dmy}{\dt} \Hv_y^{n+\half} 
= \Dlpy \DAy \frac{\Dmy}{\dt} \Hv_y^{n-\half} \\
-\Dlpy \GzEx \Dlx \Ev_x^n
+\Dlpy \GxEz \Dlz \Ev_z^n\,,
\label{eq:mtxHy}
\end{multline}
\begin{multline}
\Dlpz \DAz \frac{\Dmz}{\dt} \Hv_z^{n+\half} 
= \Dlpz \DAz \frac{\Dmz}{\dt} \Hv_z^{n-\half} \\
+\Dlpz \GyEx \Dlx \Ev_x^n
-\Dlpz \GxEy \Dly \Ev_y^n\,,
\label{eq:mtxHz}
\end{multline}
with the coefficients defined similarly to those in~\eqref{eq:mtxHx}.

From~\eqref{eq:update_Ex_internal},~\eqref{eq:update_Ex_south}, and~\eqref{eq:update_Ex_edge}, the recurrence relation for the $E_x^n$ samples can be written in matrix form as
\begin{multline}
\Dlx\DApx  \! \left( \!\frac{\Dex}{\dt}\! +\! \frac{\Dsx}{2} \!\right) \Ev_x^{n+1} 
= \Dlx\DApx \! \left(\!\frac{\Dex}{\dt} \!-\! \frac{\Dsx}{2}\!\right) \Ev_x^{n} \\
+ \Dlx\GzEx^T\Dlpy\Hv_y^{n+\half}
-\Dlx\GyEx^T\Dlpz\Hv_z^{n+\half}\\
+\Dlx\BExB \DlpyB \HyB^{n+\half} 
+\Dlx\BExN\DlpzN \HzN^{n+\half}\\
-\Dlx\BExT \DlpyT \HyT^{n+\half}
-\Dlx\BExS\DlpzS \HzS^{n+\half}\,,
\label{eq:mtxEx}
\end{multline}
where $\Dex$ and $\Dsx$ are diagonal matrices with the values of permittivity and electric conductivity, respectively, on the $x$-directed primary edges. The diagonal matrix $\DApx$ contains the area of the secondary faces where $E_x^n$ is sampled, which is equal to $\dy\dz$ for Type~1 nodes, $\dy\dz/2$ for Type~2 nodes, and $\dy\dz/4$ for Type~3 nodes. Matrix $\BExB$ consists of $0$'s and $1$'s that extract from the vector $\HyB^{n+\halftxt}$ the hanging variables corresponding to a Type~2 or Type~3 sample in $\Ev_x^n$. Matrices $\BExT$, $\BExS$, and $\BExN$ for the Top, South, and North faces serve an analogous purpose. In a similar manner, we write the matrix equations for $\Ev_y^n$ and $\Ev_z^n$.

Equations~\eqref{eq:mtxHx}--\eqref{eq:mtxHz} describing the recurrence relation for all regular magnetic field variables can be written as
\begin{equation}
\Dlp \D_{A} \frac{\D_{\mu}}{\dt} \Hv^{n+\half} = \Dlp \D_{A} \frac{\D_{\mu}}{\dt} \Hv^{n-\half}
-\Dlp \matr{C} \Dl \Ev^{n}\,,
\label{eq:mtxH}
\end{equation}
where vectors $\Ev^n$ and $\Hv^{n-\halftxt}$ collect all electric and regular magnetic field variables, respectively
\begin{equation}
\Ev^n = 
\begin{bmatrix}
\Ev_x^n\\
\Ev_y^n\\
\Ev_z^n
\end{bmatrix}\,,
\quad
\Hv^{n-\half} = 
\begin{bmatrix}
\Hv_x^{n-\half}\\
\Hv_y^{n-\half}\\
\Hv_z^{n-\half}
\end{bmatrix}\,.
\end{equation}
The diagonal matrices $\Dlp$ and $\D_{A}$ contain the secondary edge lengths and primary face areas associated with samples in $\Hv^{n-\halftxt}$. Matrix $\D_{\mu}$ contains permeability on the edges where samples in $\Hv^{n-\halftxt}$ are located. Coefficient matrix $\matr{C}$ is a discrete curl operator
\begin{equation}
\matr{C} = 
\begin{bmatrix}
\matr {0}		& 	-\GzEy	&  \phantomplus\GyEz	\\
\phantomplus\GzEx	& \matr {0}			& -\GxEz \\
-\GyEx 	&  \phantomplus\GxEy	&\matr {0}
\end{bmatrix}\,
\label{eq:C}
\end{equation}
and $\Dl$ is a diagonal matrix containing the length of the primary edges corresponding to the samples in $\Ev^n$.

Similarly, we obtain the following recurrence relation for the electric samples in the region by collecting~\eqref{eq:mtxEx} and similar equations for $\Ev_y^n$ and $\Ev_z^n$
	\begin{multline} 
	\Dl \DAp \left( \frac{\D_{\eps}}{\dt} \!+\!  \frac{\D_{\sigma}}{2}\right) \Ev^{n+1} 
	= 
	\Dl \DAp \left( \frac{\D_{\eps}}{\dt}
	\!-\!  \frac{\D_{\sigma}}{2}\right) \Ev^n \\
	+ \Dl \matr{C}^T \Dlp \Hv^{n+\half} 
	+ \Dl \Q \Sign \DlpU \vect{U}^{n+\half}\,,
	\label{eq:mtxE}
	\end{multline}
where $\DAp$ contains the secondary face areas associated with the samples in $\Ev^n$. Matrices $\D_{\eps}$, and $\D_{\sigma}$ contain permittivity and electrical conductivity on the edges that correspond to the samples in $\Ev^n$. Matrix $\Q$ contains $1$'s and $0$'s that extract the hanging variables from vector $\vect{U}^{n+\halftxt}$. For Type~1 samples in $\Ev^n$, the corresponding rows in $\Q$ contain only zeros. For a Type~2 sample in $\Ev^n$, the row of $\Q$ contains a single $1$ in the column corresponding to the hanging variable that is involved in an equation such as~\eqref{eq:update_Ex_south} for that electric field sample. For a Type~3 sample, the row of $\Q$ contains two $1$'s in the columns corresponding to the two hanging variables in the equation like~\eqref{eq:update_Ex_edge}. Matrix $\Sign$ is a diagonal matrix that selects the signs for the hanging variables according to~\eqref{eq:mtxEx}. The product $\Q \Sign$ has the following structure
\begin{equation}
\Q \Sign = 
\begin{bmatrix}
\matr{B}_{x} 	& \matr{0} 			& \matr{0} \\
\matr{0}		& \matr{B}_{y}	& \matr{0} \\
\matr{0}		& \matr{0}			& \matr{B}_{z}
\end{bmatrix}\,,
\label{eq:B_ES}
\end{equation}
where
\begin{subequations}
	\begin{eqnarray}
	&\matr{B}_{x} = \begin{bmatrix} \BExB & \BExN	&  -\BExT	& -\BExS	 \end{bmatrix} \,,& \label{eq:B_Ex}\\[2pt]
	&\matr{B}_{y} = \begin{bmatrix} \BEyW & \BEyT	&-\BEyE& -\BEyB		 \end{bmatrix} \,,&\label{eq:B_Ey}\\[2pt]
	&\matr{B}_{z} = \begin{bmatrix} \BEzS &\BEzE	&-\BEzN 	&-\BEzW  	 \end{bmatrix} \,.&\label{eq:B_Ez}
\end{eqnarray}
\end{subequations}
 Matrix $\DlpU$ in~\eqref{eq:mtxE} is a diagonal matrix containing the length of the edges where the hanging variables are sampled, namely $\dx$, $\dy$, or $\dz$ for the hanging variables collocated with Type~2 electric fields and $\dx/2$, $\dy/2$, or $\dz/2$ for the hanging variables associated with the electric fields of Type~3.

\subsection{Dynamical System Formulation}
\label{subsec:dyn_sys}

Matrix equations~\eqref{eq:mtxH} and \eqref{eq:mtxE} can be written in the form of a dynamical system as follows
\begin{subequations}
	\begin{eqnarray}
	(\R+\F) \x^{n+1}&=&(\R-\F) \x^{n} + \matr{B} \vect{u}^{n+\half}\,, \label{eq:sysA}\\
	\vect{y}^n&=&\Lm^T \x^n \,, \label{eq:sysB}
	\end{eqnarray}
\end{subequations}
where the state vector $\x^n$ contains the regular FDTD samples in the region
\begin{equation}
\x^{n} = 
\begin{bmatrix}
\Ev^n\\
\Hv^{n-\half}
\end{bmatrix}\,.
\label{eq:x}
\end{equation}
The hanging variables $U^{n+\halftxt}$ are not included in the state vector, and instead are regarded as the input of the FDTD region
\begin{equation}
\vect{u}^{n+\half} = \begin{bmatrix}
\Uv_{E_x}^{n+\half}\\[2pt]
\Uv_{E_y}^{n+\half}\\[2pt]
\Uv_{E_z}^{n+\half}
\end{bmatrix}\,,
\label{eq:u}
\end{equation}
where vector $\Uv_{E_x}^{n+\halftxt}$ collects the variables associated with the $E_x^n$ samples, namely the $U_y^{n+\halftxt}$ variables tangential to the Top and Bottom faces and $U_z^{n+\halftxt}$ variables for the North and South faces. Vectors $\Uv_{E_y}^{n+\halftxt}$ and $\Uv_{E_z}^{n+\halftxt}$ are defined similarly. The output vector $\vect{y}^n$ is defined as
\begin{equation}
\vect{y}^{n} = 
\begin{bmatrix}
\vect{Y}_{E_x}^{n}\\[1pt]
\vect{Y}_{E_y}^{n}\\[1pt]
\vect{Y}_{E_z}^{n}\\[1pt]
\end{bmatrix}\,,
\label{eq:y}
\end{equation}
where $\vect{Y}_{E_x}^{n}$ contains the $E_x^n$ samples of Types~2 and 3 associated with the hanging variables in $\Uv_{E_x}^{n+\halftxt}$. Type~2 samples are included in $\vect{Y}_{E_x}^n$ once, while each Type~3 sample is included in $\vect{Y}_{E_x}^n$ twice, at the positions of the two corresponding hanging variables in $\Uv_{E_x}^{n+\halftxt}$. Vectors $\vect{Y}_{E_y}^{n}$ and $\vect{Y}_{E_z}^{n}$ are defined similarly to $\vect{Y}_{E_x}^{n}$.

Matrices $\R$, $\F$, $\matr{B}$, and $\matr{L}^T$ in~\eqref{eq:sysA}--\eqref{eq:sysB} are given by
\begin{subequations}
\begin{eqnarray}
&\R 
=
\begin{bmatrix} 
\Dl \DAp \frac{\D_{\eps}}{\dt} & -\half \Dl \matr{C}^T \Dlp \\
-\half \Dlp \matr{C} \Dl & \Dlp \D_{A} \frac{\D_{\mu}}{\dt}
\end{bmatrix}\,, \label{eq:R}
\\
&\F = \begin{bmatrix}
\Dl \DAp \frac{\D_{\sigma}}{2} 	& -\half \Dl \matr{C}^T \Dlp \\
\half \Dlp\matr{C}\Dl				& \matr{0}
\end{bmatrix}\,,
\label{eq:F}
\end{eqnarray}
\begin{equation}
\matr{B} = \begin{bmatrix}
\Dl \Q\Sign \D_{l'_U}	\\
\matr{0}
\end{bmatrix}\,,
\label{eq:B}
\end{equation}
\begin{equation}
\Lm^T = \begin{bmatrix} \Q^T & \matr{0} \end{bmatrix} \,.
\label{eq:L}
\end{equation}
\end{subequations}

The definitions given in this section allow us to write the FDTD equations for a 3-D region in the compact matrix form~\eqref{eq:sysA}--\eqref{eq:sysB}. Remarkably, equations~\eqref{eq:sysA}--\eqref{eq:sysB} have the same structure as in the 2-D case~\cite{jnl-2017-fdtd-dissipative}. This fact allows us to generalize previous results valid in two dimensions to the most general 3-D case.

\section{Dissipativity of a 3-D FDTD Region}
\label{sec:fdtd_3d_dissipativity}

In this section, we derive dissipativity conditions for the 3-D FDTD system~\eqref{eq:sysA}--\eqref{eq:sysB} defined in Sec.~\ref{sec:dyn_sys} and explain their physical significance. 
\begin{definition}
According to the theory of dissipative systems~\cite{byrnes1994losslessness}, the discrete-time dynamical system~\eqref{eq:sysA}--\eqref{eq:sysB} is dissipative with the supply rate $s\left(\vect{y}^{n},\vect{u}^{n+\halftxt}\right)$ if we can find a function $\E$ that satisfies
\begin{subequations}
\begin{align}
	\storage{\x^{n}} &\ge 0 \,, \quad \forall \x^n \,,
\\
	\storage{\vect{0}} &= 0  \,,
	\\
	\phantom{sps}	\storage{\vect{x}^{n+1}} - \storage{\vect{x}^n} &\le \supply{\vect{y}^n, \vect{u}^{n+\half}} \,,  \, \forall \vect{u}^{n+\half}\, \forall n \,.
	\label{eq:dissipation_ineq}
\end{align}
	\label{eq:dissipation}
\end{subequations}
\end{definition}

The storage function $\storage{\x^n}$ quantifies the energy stored in the FDTD region and the supply rate $\supply{\vect{y}^n, \vect{u}^{n+\halftxt}}$ is the energy absorbed by the region through its boundaries between time points $n$ and $n\!+\!1$.

We define the storage function and the supply rate as
\begin{equation}
\storage{\vect{x}^n} = \frac{\dt}{2} \left (\vect{x}^n \right) ^T \matr{R} \vect{x}^n\,,
\label{eq:storage1}
\end{equation}
\begin{equation}
\supply{\vect{y}^n, \vect{u}^{n+\half}} = \dt \frac{\left( \vect{y}^n + \vect{y}^{n+1} \right)}{2}^T  \IO \vect{u}^{n+\half}\,.
\label{eq:supply1}
\end{equation}
The diagonal matrix $\IO$ in~\eqref{eq:supply1} is given by
\begin{equation}
\IO = \D_{l_Y}  \DlpU \Sign \,,
\label{eq:def_S}
\end{equation}
where $\D_{l_Y}$ contains the lengths of the primary edges on which the electric samples in $\vect{y}^n$ are taken. Expressions~\eqref{eq:storage1} and~\eqref{eq:supply1} generalize those from the 2-D case~\cite{jnl-2017-fdtd-dissipative}.

\subsection{Physical Meaning of the Supply Rate and Storage Function}

By substituting~\eqref{eq:x} and~\eqref{eq:R} into~\eqref{eq:storage1}, we can write the storage function as
\begin{multline}
\storage{\vect{x}^n} = 
\frac{1}{2} (\Ev^n)^T \Dl \DAp \D_{\eps} \Ev^n\\
+ \frac{\dt}{2} \left(\Hv^{n-\half}\right)^T  \left[ \Dlp\D_{A} \frac{\D_{\mu}}{\dt}\Hv^{n-\half}
- \Dlp \matr{C} \Dl \Ev^n \right] \,,
\label{eq:storage1p5}
\end{multline}
which, using~\eqref{eq:mtxH} to simplify the term inside the square brackets, reduces to
\begin{multline}
\storage{\vect{x}^n} = 
\frac{1}{2} (\Ev^n)^T \Dl \DAp \D_{\eps} \Ev^n\\
+ \frac{1}{2} \left(\Hv^{n-\half}\right)^T \Dlp\D_{A} \D_{\mu}\Hv^{n+\half}\,.
\label{eq:storage2}
\end{multline}
This expression for the stored energy has been proposed in~\cite{edelvik2004general} for FDTD/FIT. Equation~\eqref{eq:storage2} reveals the physical meaning of the storage function~\eqref{eq:storage1} as a discrete analogy of 
\begin{equation}
\half \iiint_{V} \! \left( \eps_x E_x^2  \!+ \eps_y E_y^2 \!+  \eps_z E_z^2  \!+ 
\mu_x H_x^2  \!+  \mu_y H_y^2 \!+   \mu_z H_z^2 \right) dV \,,
\label{eq:energy_continuous}
\end{equation}
which is the energy stored in electric and magnetic fields inside a given volume.

For the supply rate, we can substitute~\eqref{eq:def_S} into~\eqref{eq:supply1}, obtaining
\begin{equation}
s(\vect{y}^n, \vect{u}^{n+\half}) = \dt \frac{\left( \vect{y}^n + \vect{y}^{n+1} \right)}{2}^T \D_{l_Y}  \DlpU \Sign \vect{u}^{n+\half}\,.
\label{eq:supply2}
\end{equation}
The product $\D_{l_Y}  \DlpU$ is a diagonal matrix containing the area of boundary faces associated with each hanging variable and the corresponding electric field sample. Diagonal matrix P has a ${+1}$ on the faces where the Poynting vector points into the region and a ${-1}$ otherwise. As a result, \eqref{eq:supply2} is a discrete version of the Poynting integral over the region's boundary from time sample $n$ to $n\!+\!1$ 
\begin{equation}
\int_{n \dt}^{(n+1)\dt}  \oiint_{S} \vec{E} \times \vec{H} \cdot \hat{n} \ dA \ dt \,,
\end{equation}
where $\hat{n}$ is the inward unit normal vector. Therefore, supply rate~\eqref{eq:supply1} is the total energy that enters the region from its boundaries between time $n$ and time $n\!+\!1$.

\subsection{Dissipativity Conditions for a 3-D FDTD Region}

By substituting~\eqref{eq:storage1} and~\eqref{eq:supply1} into~\eqref{eq:dissipation_ineq}, we arrive at the following theorem, which gives some simple conditions for the dissipativity of~\eqref{eq:sysA}--\eqref{eq:sysB}.

\begin{theorem}
	\label{theorem:dissipativity}
	A system in the form~\eqref{eq:sysA}--\eqref{eq:sysB} is dissipative with $\storage{\vect{x}^n}$ in~\eqref{eq:storage1} as storage function and $\supply{\vect{y}^n, \vect{u}^{n+\halftxt}}$ in~\eqref{eq:supply1} as supply rate if
	\begin{subequations}
		\begin{align}
		 \matr{R} = \matr{R}^T &>  0\,,  \label{eq:dissip_cond1} \\
		 \matr{F} + \matr{F}^T &\ge 0\,, \label{eq:dissip_cond2}  \\
		\Lm \IO &=  \matr{B} \,.  \label{eq:dissip_cond3} 
		\end{align}
		\label{eq:dissip_cond}
	\end{subequations}
\end{theorem}
\begin{proof}
	See~\cite{jnl-2017-fdtd-dissipative}.
\end{proof}

The physical meaning of the first condition~\eqref{eq:dissip_cond1} can be better understood by applying the Schur complement~\cite{Boy94} to transform~\eqref{eq:dissip_cond1} into
\begin{equation}
\begin{cases}
\Dl \DAp \frac{\D_{\eps}}{\dt} > 0\\[6pt]
\Dlp \D_{A} \frac{\D_{\mu}}{\dt} - \frac{\dt}{4} \Dlp \matr{C} \Dl \D_{\eps}^{-1}\DAp^{-1}  \matr{C}^T \Dlp > 0 
\end{cases}\,.
\label{eq:R_Sch}
\end{equation}
The first matrix inequality in~\eqref{eq:R_Sch} is true because lengths, areas, and permittivities in the region are positive. Applying a congruence with matrix $2(\dt\Dlp\DA \Dm)^{-\halftxt}$ to the left hand side of the second inequality in~\eqref{eq:R_Sch}, we obtain an equivalent condition to~\eqref{eq:R_Sch}:
\begin{equation}
\frac{4}{\dt^2}\I_{N_H}- \Singvals^T \Singvals > 0\,,
\label{eq:R_Sch_b_2}
\end{equation}
where
\begin{equation}
\Singvals = \Dl^{\half}\D_{\eps}^{-\half} \DAp^{-\half} \matr{C}^T\Dlp^{\half}\Dm^{-\half}\DA^{-\half}  \,.
\end{equation}
Condition~\eqref{eq:R_Sch_b_2} holds when
\begin{equation}
\dt < \min_k{\left\{\frac{2}{s_k} \right\}}\,,
\label{eq:dissip_cond_sing_vals}
\end{equation}
where $s_k$ are the non-zero singular values of $\Singvals$. Thus, condition~\eqref{eq:dissip_cond1} can be seen as a generalized CFL limit, since it is applicable to the most general case of a lossy region with nonuniform material properties. Moreover, with a strategy similar to~\cite{edelvik2004general} we can show that a sufficient condition for~\eqref{eq:dissip_cond1} to hold is that the classical FDTD CFL limit~\cite{Gedney} is met
\begin{equation}
\dt < \frac{\sqrt{\mu \eps}}{\sqrt{ \frac{1}{\dx^2} + \frac{1}{\dy^2} + \frac{1}{\dz^2} }}\,,
\label{eq:CFL}
\end{equation}
where $\eps$ is the smallest primary edge permittivity in the region and $\mu$ is the smallest secondary edge permeability. When the CFL limit is violated, the energy stored in some cells is no longer bounded below by zero, which can make them capable of supplying unlimited energy to the rest of the system, potentially causing instability.
 
Condition~\eqref{eq:dissip_cond2} expands into
\begin{equation}
\F + \F^T=
\begin{bmatrix} 
\Dl \DAp \D_{\sigma} 	& \matr{0} \\
\matr{0}			& \matr{0}
\end{bmatrix} \ge 0\,,
\end{equation}
which holds when the conductivity on each edge is non-negative. 

The third condition~\eqref{eq:dissip_cond3} is always true. In order to see this, we expand $\Lm \IO$ as
\begin{equation}
\Lm \IO = \begin{bmatrix}\Q \D_{l_Y} \Sign \DlpU  \\ \matr{0} \end{bmatrix}\,.
\end{equation}
Right-multiplication by $\D_{l_Y}$ has the same effect on $\Q$ as left-multiplication by $\Dl$. Hence, $\Q \D_{l_Y} =  \Dl \Q$, and as a result $\Lm \IO = \matr {B}$.

In summary, Theorem~\eqref{theorem:dissipativity} shows that the FDTD system~\eqref{eq:sysA}--\eqref{eq:sysB} associated to an arbitrary region is dissipative if
\begin{enumerate}
	\item all conductivities are non-negative, as one may obviously expect;
	\item the CFL limit~\eqref{eq:dissip_cond1} is respected.
\end{enumerate}
If these two conditions are satisfied, the FDTD region can be arbitrarily interconnected with any other dissipative subsystem without violating stability.

\section{Systematic Method for Stability Enforcement}
\label{sec:method}

The developments in the previous sections provide a powerful approach to construct both simple and advanced FDTD schemes with guaranteed stability. The method is general since it is applicable to complicated setups involving multiple subsystems, such as FDTD grids with different resolution, multiple boundary conditions, circuit models, or reduced-order models. The approach involves the following steps:
\begin{enumerate}
	\item Hanging variables, which are seen as inputs, are defined at the boundary of each block. The corresponding electric field samples are regarded as outputs, as done in Sec.~\ref{sec:dyn_sys} for a 3D-FDTD region.
	\item If subsystems cannot be connected directly, for example because of different grid resolution, an interpolation rule is created, and viewed as a separate subsystem.
	\item Stability is enforced by ensuring that all subsystems are dissipative. For dynamical systems of the form~\eqref{eq:sysA}--\eqref{eq:sysB} this can be done using~\eqref{eq:dissip_cond1}--\eqref{eq:dissip_cond3}. Since the connection of dissipative systems is also dissipative~\cite{jnl-2016-tap-fdtdmorstable}, the resulting scheme is stable by construction.
	\item The most restrictive time step is taken to ensure that all subsystems are dissipative.
\end{enumerate}

The proposed approach significantly simplifies the creation of new FDTD schemes with guaranteed stability, since dissipativity conditions can be imposed on individual subsystems. In contrast, existing techniques to analyze and enforce stability, such as the energy~\cite{edelvik2004general} and iteration~\cite{Gedney, Remis} methods, require the analysis of the entire scheme, which can be a formidable task. A remarkable feature of the proposed approach is its modularity. Given a set of subsystems that have been proven to be dissipative, any of their combination is guaranteed to be stable. This feature is hoped to accelerate scientific research in the FDTD area, since it allows researchers to create new stable schemes without having to redo stability analysis for every change. The proposed modular framework is relevant also for commercial FDTD solvers, where users want to be able to combine available FDTD subsystems in the way most suitable to model their problem, while having an absolute guarantee of stability.

It should be noted that, although we use a matrix formulation to investigate the dissipativity of FDTD equations, the final scheme can be implemented with scalar FDTD equations for optimal efficiency.

\section{Application to Stable 3-D FDTD Subgridding}
\label{sec:subgridding}

As a demonstration of the stability framework in Sec.~\ref{sec:method}, we derive a stable subgridding algorithm for 3-D FDTD. In the proposed algorithm, a coarse grid is refined in selected regions to better resolve fine geometrical details. We denote the refinement factors as $r_x$, $r_y$, and $r_z$ in the $x$, $y$, and $z$ directions, respectively. The refinement factors are odd integers greater than one. 

For simplicity, we present the proposed method by considering the scenario in Fig.~\ref{fig:system_view}, where the interface between the coarse and fine grids is planar. In the Appendix, we discuss how corners can be similarly treated.

In order to develop a stable subgridding scheme with the method in Sec.~\ref{sec:method}, we view a subgridding algorithm as a connection of four subsystems: boundary conditions, the coarse grid, the fine grid, and the interpolation rule that relates fields at the interface of the two grids, as depicted in Fig.~\ref{fig:system_view}.

\begin{figure}[t]
	\centering
	\includegraphics[scale=0.95]{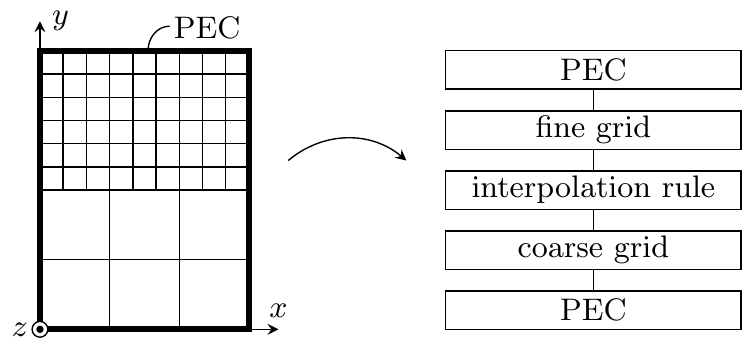}
	\caption{Illustration of the subgridding scenario. Left: $xy$ cross section of the primary grid. Right: subsystem interpretation.}
	\label{fig:system_view}
\end{figure}

\subsection{Interpolation Conditions}
\label{subsec:subgridding_interp}

We first describe the interpolation rule established between the coarse and fine grid fields, which is the core of every subgridding algorithm. In the scenario of Fig.~\ref{fig:system_view}, the interpolation rule has to properly relate the fields tangential to the North boundary of the coarse grid to the fields tangential to the South boundary of the fine grid. We discuss the interpolation rule in detail for the $E_z^n$-$U_x^{n+\halftxt}$ pairs. The derivations for $E_x^n$-$U_z^{n+\halftxt}$ pairs can be done analogously.

As shown in the left panel of Fig.~\ref{fig:face_interp}, we consider the portion of the interface surrounding one coarse electric field sample (denoted as $E_z^n$) and the corresponding hanging variable (denoted as $U_x^{n+\halftxt}$). The fine samples that fall in the same region in the case of $r_x\!=r_z\!=3$ are shown in the right panel of Fig.~\ref{fig:face_interp}.

The fine samples are numbered form $1$ to $9$, as the nodes where they are sampled. The fine electric fields $\hat{E}_z^n$ are set equal in the $z$ direction 
\begin{subequations}
	\begin{eqnarray}
	\Ehat_{z1}^n = \Ehat_{z4}^n = \Ehat_{z7}^n \,, &\quad \forall n\,, \label{eq:face_E_eq1} \\
	\Ehat_{z2}^n = \Ehat_{z5}^n = \Ehat_{z8}^n \,, &\quad \forall n\,, \label{eq:face_E_eq2} \\
	\Ehat_{z3}^n = \Ehat_{z6}^n = \Ehat_{z9}^n   \,,& \quad \forall n\,,\label{eq:face_E_eq3}
	\end{eqnarray}
\label{eq:face_E_eq}
\end{subequations}
and the fine hanging variables $\hat{U}_x^{n+\halftxt}$ are set equal in the $x$ direction
\begin{subequations}
	\begin{eqnarray}
	\Uhat_{x1}^{n+\half} = \Uhat_{x2}^{n+\half} = \Uhat_{x3}^{n+\half} \,, &\quad \forall n\,,  \\
	\Uhat_{x4}^{n+\half} = \Uhat_{x5}^{n+\half} = \Uhat_{x6}^{n+\half} \,, &\quad \forall n\,, \\
	\Uhat_{x7}^{n+\half} = \Uhat_{x8}^{n+\half} = \Uhat_{x9}^{n+\half} \,, &\quad \forall n\,,
	\end{eqnarray}
\label{eq:face_U_eq}
\end{subequations}
where the ``$\hat{\phantom{a}}$'' notation refers to the variables of the refined region.

The coarse samples $E_z^n$ and $U_x^{n+\halftxt}$ are forced to be equal to the average of the fine samples through the following interpolation rules
\begin{eqnarray}
&E_z^n = \frac{1}{r_x} \T_{r_x}^T \Evhat_z^n \,,
\quad &\forall n \,,
\label{eq:face_aver_E}
\\
&U_x^{n+\half} = \frac{1}{r_z} \T_{r_z}^T \Uvhat_x^{n+\half}\,,
\quad &\forall n\,,
\label{eq:face_aver_U}
\end{eqnarray}
where
\begin{equation}
\Evhat_z^n = 
\begin{bmatrix}
\Ehat_{z1}^n \\[2pt]
\Ehat_{z2}^n \\[2pt]
\Ehat_{z3}^n\\[2pt]
\end{bmatrix}\,, 
\quad
\Uvhat_x^{n+\half}=
\begin{bmatrix}
\Uhat_{x1}^{n+\half}\\[2pt]
\Uhat_{x4}^{n+\half}\\[2pt]
\Uhat_{x7}^{n+\half}\\[2pt]
\end{bmatrix}\,\label{eq:def_face}
\end{equation}
and $\T_m$ is an $m\times1$ matrix of ones. 

Interpolation rules \eqref{eq:face_E_eq}, \eqref{eq:face_U_eq}, \eqref{eq:face_aver_E}, and \eqref{eq:face_aver_U} draw inspiration from reciprocity theory for stable subgridding~\cite{Reciprocity}. In the $z$ direction, the pair of rules~\eqref{eq:face_E_eq} and~\eqref{eq:face_aver_U} is similar to~\cite{Reciprocity}, except the magnetic fields are related directly at the interface, as opposed to a plane that is offset from the interface as in~\cite{Reciprocity}. In the $x$ direction, the pair~\eqref{eq:face_U_eq} and~\eqref{eq:face_aver_E} is analogous, except the magnetic and electric interpolation rules are switched.

\begin{figure}[t]
	\centering
	 \includegraphics[scale=0.95]{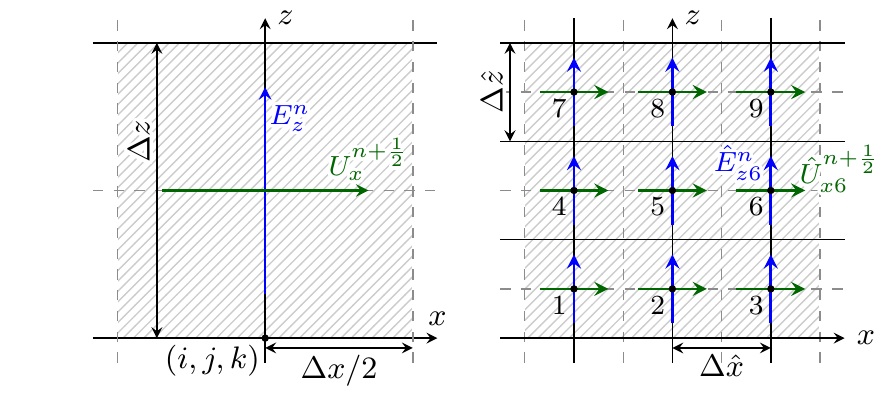}
	\caption{Coarse and fine sides of the subgridding interface in Sec.~\ref{sec:subgridding}.}
	\label{fig:face_interp}
\end{figure}

\subsection{Stability Proof}
\label{subsec:subgridding_stability}

The stability of the proposed subgridding algorithm can be proved with the dissipation theory in Sec.~\ref{sec:method}. As shown in Fig.~\ref{fig:system_view}, the proposed scheme can be seen as the connection of various subsystems, which must all be dissipative in order to ensure stability. As proven in Sec.~\ref{sec:fdtd_3d_dissipativity}, the coarse and fine grids are dissipative under their respective CFL limits. The perfect electric conductor (PEC) boundary condition is lossless, since imposing a null electric field means that no energy will be exchanged at any time. 

What is left to investigate is whether the interpolation rule is dissipative or not. We show that the interpolation rule in Sec.~\ref{subsec:subgridding_interp} is actually lossless for any time step, by proving that the supply rate~\eqref{eq:supply2} of the interpolation rule is zero. The supply rate of the interpolation rule is 
the summation of terms, each one associated to one hanging 
variable on its boundaries. The terms associated with the $x$-directed 
hanging variables $U_x^{n+\halftxt}$ and $\hat{\vect{U}}_x^{n+\halftxt}$ read
\begin{align}
&s_{U_x}^{n+\half} = +\dt \dx \dz \frac{E_z^n + E_z^{n+1}}{2}  U_x^{n+\half} \label{eq:supply_face0}\\
&- \dt \frac{\dx}{r_x} \frac{\dz}{r_z} \frac{\left(
	\T_{r_z} \!\! \otimes\Evhat_z^{n}  +  \T_{r_z} \!\!\otimes\Evhat_z^{n+1}  
	\right)}{2}^{\!T}  \!\!\!
\left(\Uvhat_x^{n+\half} \!\otimes\T_{r_x} \right)\,. \nonumber
\end{align}
Symbol ``$\otimes$'' in~\eqref{eq:supply_face0} denotes the Kronecker product~\cite{steeb}. The supply rate term~\eqref{eq:supply_face0} accounts for the contribution of the $U_x^{n+\halftxt}$ and $\hat{U}_x^{n+\halftxt}$ samples in Fig.~\ref{fig:face_interp} and the discussion for $z$-directed hanging variables is analogous. Using the properties of the Kronecker product~\cite{steeb}, we can rewrite~\eqref{eq:supply_face0} as
\begin{multline}
s_{U_x}^{n+\half} = +\dt \dx \dz \frac{E_z^n + E_z^{n+1}}{2}  U_x^{n+\half}\\
- \dt \frac{\dx}{r_x} \frac{\dz}{r_z} \T_{r_z}^T \Uvhat_x^{n+\half} \frac{\left(\Evhat_z^{n} + \Evhat_z^{n+1}\right)}{2}^{\!T} \T_{r_x} \,.
\label{eq:supply_face}
\end{multline}
Substituting~\eqref{eq:face_aver_E} and~\eqref{eq:face_aver_U} into~\eqref{eq:supply_face}, we see that the terms corresponding to the coarse and fine grids cancel each other
\begin{multline}
s_{U_x}^{n+\half} = +\dt \dx \dz \frac{E_z^n + E_z^{n+1}}{2}  U_x^{n+\half}\\
- \dt \frac{\dx}{r_x} \frac{\dz}{r_z} r_z U_x^{n+\half} \frac{\left(r_x E_z^{n} + r_x E_z^{n+1}\right)}{2}  =0\,.
\label{eq:supply_face2}
\end{multline}

Iterating the argument for all other coarse hanging variables and the corresponding groups of $r_x r_z$ hanging variables on the fine side, we conclude that the total supply rate to the interpolation rule is zero. Hence, the energy leaving the fine grid during a time step is equal to the energy that the interpolation rule supplies to the coarse grid. This makes the interpolation rule a lossless, and hence dissipative, subsystem. Since the coarse and fine grids are dissipative under their own CFL limits, the entire scheme is stable under the CFL limit of the fine grid, which is the most restrictive CFL limit.

\subsection{Practical Implementation}
\label{subsec:subgridding_update}

In this section we discuss how the proposed subgridding algorithm can be implemented. For updating all field samples strictly inside the two grids, one can use conventional FDTD equations. The update equation for the field samples on the refinement interface is instead derived from the interpolation rules~\eqref{eq:face_E_eq}, \eqref{eq:face_U_eq}, \eqref{eq:face_aver_E}, and~\eqref{eq:face_aver_U}.

As discussed in Sec.~\ref{subsec:eq_each_node}, each electric field sample $E_z^n$ on the North boundary of the coarse grid satisfies
\begin{multline}
\dx \dyhalf \left( \frac{\eps_z}{\dt} \!+\! \frac{\sigma_z}{2} \right) E_z^{n+1} \\
= \dx \dyhalf \left( \frac{\eps_z}{\dt} \!-\! \frac{\sigma_z}{2} \right) E_z^n 
- \dx U_x^{n+\half}\\
- \dyhalf  H_y\big|_{i-\half}^{n+\half} 
+ \dx H_x\big|_{j-\half}^{n+\half}
+ \dyhalf H_y\big|_{i+\half}^{n+\half} 
\,,
\label{eq:face_crs}
\end{multline}
where only those subscripts that are different from $i$, $j$, and $k\!+\!\halftxt$ are explicitly shown for clarity. Each electric field sample at a fine node ${\mhat \in \{1, 2,\dots, 9\}}$ with coordinates ${(\ihat, \jhat=\!1, \khat\!+\!\halftxt)}$ satisfies
\begin{multline}
\frac{\dx}{r_x} \frac{\dy}{2r_y} \left(\frac{\hat{\eps}_{z\mhat}}{\dt} \!+\! \frac{\hat{\sigma}_{z\mhat}}{2} \right) \hat{E}_{z\mhat}^{n+1} 
\\= \frac{\dx}{r_x} \frac{\dy}{2r_y} \left(\frac{\hat{\eps}_{z\mhat}}{\dt} \!-\! \frac{\hat{\sigma}_{z\mhat}}{2} \right) \hat{E}_{z\mhat}^{n} 
+ \frac{\dx}{r_x}  \hat{U}_{x\mhat}^{n+\half}\\
+ \frac{\dy}{2r_y} \hat{H}_{y\mhat}\big|^{n+\half}_{\ihat+\half} 
-\frac{\dx}{r_x} \hat{H}_{x\mhat}\big|^{n+\half}_{\jhat+\half} 
- \frac{\dy}{2r_y} \hat{H}_{y\mhat}\big|^{n+\half}_{\ihat-\half}
\,,
\label{eq:face_fine}
\end{multline}
where a sample $\hat{H}_{y\mhat}\big|^{n+\halftxt}_{\ihat+\halftxt}$ has the same coordinates as the node $\mhat$, except for the $x$ coordinate, which is $\ihat\!+\!\halftxt$. Similar notation is used for the other samples in~\eqref{eq:face_fine}.

Using~\eqref{eq:face_E_eq1} and averaging~\eqref{eq:face_fine} over the $\khat$ indexes in the shaded area of Fig.~\ref{fig:face_interp}, we obtain the following equation for the fine sample $\Ehat_{z1}^n$
\begin{multline}
\frac{\dx}{r_x} \frac{\dy}{2r_y} 
\frac{1}{r_z}\sum_{\mhat = 1,4,7}\! \left(
 \frac{\hat{\eps}_{z\mhat}}{\dt} \!+\! \frac{\hat{\sigma}_{z\mhat}}{2}  
\right)
\hat{E}_{z1}^{n+1} 
\\= \frac{\dx}{r_x} \frac{\dy}{2r_y} 
\frac{1}{r_z}\sum_{\mhat = 1,4,7}\! \left(
 \frac{\hat{\eps}_{z\mhat}}{\dt} \!-\! \frac{\hat{\sigma}_{z\mhat}}{2}  
\right) \hat{E}_{z1}^{n} \\
+ \frac{\dx}{r_x} \frac{1}{r_z}\sum_{\mhat = 1,4,7} \!\!\hat{U}_{x\mhat}^{n+\half}
+ \frac{\dy}{2r_y} \frac{1}{r_z}\sum_{\mhat = 1,4,7}\!\!\hat{H}_{y\mhat}\big|^{n+\half}_{\ihat+\half}\\
-\frac{\dx}{r_x} \frac{1}{r_z}\sum_{\mhat = 1,4,7} \!\!\hat{H}_{x\mhat}\big|^{n+\half}_{\jhat+\half} 
- \frac{\dy}{2r_y} \frac{1}{r_z}\sum_{\mhat = 1,4,7} \!\!\hat{H}_{y\mhat}\big|^{n+\half}_{\ihat-\half}
 \,,
\label{eq:face_fine_explicit_average0}
\end{multline}
and similarly for the other samples in $\Evhat_z^n$, namely for $\Ehat^n_{z2}$ and $\Ehat_{z3}^n$. Next, we use~\eqref{eq:face_aver_U} to replace the average of the fine hanging variables in~\eqref{eq:face_fine_explicit_average0} with $U_x^{n+\halftxt}$
\begin{multline}
\frac{\dx}{r_x} \frac{\dy}{2r_y} 
\frac{1}{r_z}\sum_{\mhat = 1,4,7}\! \left(
\frac{\hat{\eps}_{z\mhat}}{\dt} \!+\! \frac{\hat{\sigma}_{z\mhat}}{2}  
\right)
\hat{E}_{z1}^{n+1} 
\\= \frac{\dx}{r_x} \frac{\dy}{2r_y} 
\frac{1}{r_z}\sum_{\mhat = 1,4,7}\! \left(
\frac{\hat{\eps}_{z\mhat}}{\dt} \!-\! \frac{\hat{\sigma}_{z\mhat}}{2}  
\right) \hat{E}_{z1}^{n} 
+ \frac{\dx}{r_x} {U}_{x}^{n+\half}\\
+ \frac{\dy}{2r_y} \frac{1}{r_z}\sum_{\mhat = 1,4,7}\!\!\hat{H}_{y\mhat}\big|^{n+\half}_{\ihat+\half}
-\frac{\dx}{r_x} \frac{1}{r_z}\sum_{\mhat = 1,4,7} \!\!\hat{H}_{x\mhat}\big|^{n+\half}_{\jhat+\half} \\
- \frac{\dy}{2r_y} \frac{1}{r_z}\sum_{\mhat = 1,4,7} \!\!\hat{H}_{y\mhat}\big|^{n+\half}_{\ihat-\half}
\,.
\label{eq:face_fine_explicit_average}
\end{multline}
Because of~\eqref{eq:face_U_eq}, the same substitution can be done in the equations for $\Ehat^n_{z2}$ and $\Ehat_{z3}^n$.

Writing~\eqref{eq:face_fine_explicit_average} and the equations for $\Ehat^n_{z2}$ and $\Ehat_{z3}^n$ in matrix form, we obtain the following relation for the $\Evhat^n$ vector
\begin{multline}
\frac{\dx}{r_x} \frac{\dy}{2r_y}  \left(\frac{\Dhat_{\eps_z}}{\dt} \!+\! \frac{\Dhat_{\sigma_z}}{2} \right) \Evhat_z^{n+1} 
\\= \frac{\dx}{r_x} \frac{\dy}{2r_y}  \left(\frac{\Dhat_{\eps_z}}{\dt} \!-\! \frac{\Dhat_{\sigma_z}}{2} \right) \Evhat_z^{n}  
+ \frac{\dx}{r_x} \T_{r_x} U_x^{n+\half}\\
+ \frac{\dy}{2r_y}  \hat{\Hv}^{n+\half}_{y,\ihat+\half} 
 -\frac{\dx}{r_x}  \hat{\Hv}_{x, \jhat+\half}^{n+\half}
 - \frac{\dy}{2r_y}  \hat{\Hv}^{n+\half}_{y,\ihat-\half}
\,,
\label{eq:face_fine_mtx}
\end{multline}
where
\begin{subequations}
	\begin{eqnarray}
	\Hvhat_{x, \jhat+\half}^{n+\half} &=&
	\frac{1}{r_z}
	\begin{bmatrix}
	 \sum\limits_{\mhat=1,4,7} \hat{H}_{x\mhat}\big|^{n+\half}_{\jhat+\half}\\[6pt]
	 \sum\limits_{\mhat=2,5,8} \hat{H}_{x\mhat}\big|^{n+\half}_{\jhat+\half}\\[6pt]
	\sum\limits_{\mhat=3,6,9} \hat{H}_{x\mhat}\big|^{n+\half}_{\jhat+\half}\\[6pt]
	\end{bmatrix}\,,\label{eq:def_Hx_face}
	\\
	\Hvhat_{y,\ihat-\half}^{n+\half} &=&
	\frac{1}{r_z} 
	\begin{bmatrix}
	\sum\limits_{\mhat=1,4,7}
	\hat{H}_{y\mhat}\big|^{n+\half}_{\ihat-\half} \\[6pt]
	 \sum\limits_{\mhat=2,5,8}
	\hat{H}_{y\mhat}\big|^{n+\half}_{\ihat-\half}\\[6pt]
	 \sum\limits_{\mhat=3,6,9}
	\hat{H}_{y\mhat}\big|^{n+\half}_{\ihat-\half}\\[6pt]
	\end{bmatrix}\,,\label{eq:def_Hy_i_minus_face} 
	\end{eqnarray}
\end{subequations}
and similarly for $\Hvhat_{y,\ihat+\halftxt}^{n+\halftxt}$. \color{black}
Diagonal matrix $\Dhat_{\eps_z}$ in~\eqref{eq:face_fine_mtx} contains the permittivity of the half-cells adjacent to the fine grid's South boundary, averaged in a way similar to the magnetic fields in~\eqref{eq:def_Hx_face}--\eqref{eq:def_Hy_i_minus_face}. The conductivity matrix $\Dhat_{\sigma_z}$ is defined analogously.

Substituting~\eqref{eq:face_aver_E} into~\eqref{eq:face_crs} and multiplying the result on the left by $\T_{r_x}/r_x$, we obtain
\begin{multline}
\frac{\dx}{r_x} \dyhalf \left( \frac{\eps_z}{\dt} + \frac{\sigma_z}{2} \right) \frac{\one_{r_x} }{r_x}\Evhat_z^{n+1} \\
= \frac{\dx}{r_x} \dyhalf \left( \frac{\eps_z}{\dt} - \frac{\sigma_z}{2} \right)\frac{\one_{r_x} }{r_x}\Evhat_z^{n}
- \dx \frac{\T_{r_x}}{r_x} U_x^{n+\half}\\
- \dyhalf \frac{\T_{r_x}}{r_x} H_y\big|_{i-\half}^{n+\half} 
+ \dx\frac{\T_{r_x}}{r_x}   H_x\big|_{j-\half}^{n+\half}
+ \dyhalf \frac{\T_{r_x}}{r_x} H_y\big|_{i+\half}^{n+\half} 
\,,
\label{eq:face_crs2}
\end{multline}
where $\one_{\rx}$ denotes a square matrix of ones of size ${\rx \!\times \rx}$. Adding~\eqref{eq:face_fine_mtx} to~\eqref{eq:face_crs2} to eliminate $U_x^{n+\halftxt}$, we obtain an ${\rx \!\times \rx}$ matrix equation
\begin{multline}
\frac{\dx}{r_x}
\left(\frac{\dy}{2} + \frac{\dy}{2r_y} \right)
\left(  \frac{\matr {M}_{\eps_z}}{\dt} +  \frac{\matr {M}_{\sigma_z}}{2}\right)
\Evhat_z^{n+1} \\
=  
\frac{\dx}{r_x}
\left(\frac{\dy}{2} + \frac{\dy}{2r_y} \right)
\left(  \frac{\matr {M}_{\eps_z}}{\dt} -  \frac{\matr {M}_{\sigma_z}}{2}\right)
\Evhat_z^{n}\\
- \dyhalf \frac{\T_{r_x}}{r_x} H_y\big|_{i-\half}^{n+\half} 
+ \dx\frac{\T_{r_x}}{r_x}   H_x\big|_{j-\half}^{n+\half}
+ \dyhalf \frac{\T_{r_x}}{r_x} H_y\big|_{i+\half}^{n+\half}  \\
+ \frac{\dy}{2r_y}  \hat{\Hv}^{n+\half}_{y,\ihat+\half} 
-\frac{\dx}{r_x}  \hat{\Hv}_{x, \jhat+\half}^{n+\half}
- \frac{\dy}{2r_y}  \hat{\Hv}^{n+\half}_{y,\ihat-\half} \,,
\label{eq:face_update}
\end{multline}
with the vector of unknowns $\Evhat_z^{n+1}$, where matrices $\matr {M}_{\eps_z}$ and $\matr {M}_{\sigma_z}$ are given by
\begin{align}
\matr {M}_{\eps_z} &=&
\left(\frac{\dy}{2} + \frac{\dy}{2r_y} \right)^{-1}
\left( \dyhalf \frac{\one_{r_x} }{r_x}\eps_z + 
\frac{\dy}{2\ry}\Dhat_{\eps_z}
\right)\,,
\\
\matr {M}_{\sigma_z} &=& \left(\frac{\dy}{2} + \frac{\dy}{2r_y} \right)^{-1} \left(
\dyhalf \frac{\one_{r_x} }{r_x}\sigma_z
+  \frac{\dy}{2\ry}\Dhat_{\sigma_z}\right)  \,.
\end{align}
Equation~\eqref{eq:face_update} relates the electric field vector $\Evhat_z$ to the surrounding magnetic fields in both the coarse and the fine grids. We can see that this last step eliminates the hanging variables on the boundaries of the two grids, which serve as a temporary means to connect the two grids. Equation~\eqref{eq:face_update} is a recursive relation for the electric field samples at the interface, that can be used to compute $\Evhat_z^{n+1}$ knowing the electric and magnetic samples at the previous times. Although~\eqref{eq:face_update} is not fully explicit in terms of $\Evhat_z^{n+1}$, the matrix in front of this unknown is very small, having size ${r_x \times r_x}$. Moreover, it is a constant. Therefore, this matrix has to be inverted only once, leading to  an explicit update equation for $\Evhat_z^{n+1}$, which will be used for all time steps.
 
Overall, the proposed subgridding scheme consists of the following steps.
\begin{enumerate}
	\item Starting from $E^n$ and $H^{n+\halftxt}$, we update all electric fields strictly inside fine and coarse regions using standard FDTD equations such as~\eqref{eq:update_Ex_internal}.
	\item We update the electric fields on the interface between the two grids using~\eqref{eq:face_update} and a similar equation for $\Ev_x^n$ samples. Once~\eqref{eq:face_update} is solved, the fine electric samples $\Ehat_{z4}^{n+1}$,~$\hdots$~,~$\Ehat_{z9}^{n+1}$ can be found using the equalities~\eqref{eq:face_E_eq1}--\eqref{eq:face_E_eq3} and the coarse electric field sample can be updated using~\eqref{eq:face_aver_E}.
	\item All regular magnetic field samples are updated using standard FDTD equations, such as~\eqref{eq:update_Hx_internal} and~\eqref{eq:update_Hx_face}, to obtain the values at $n\!+\!(3/2)$.
\end{enumerate} 

The inclusion of material properties throughout the derivations rigorously guarantees stability in the case when arbitrary permittivities and conductivities are assigned to cells on each side of the subgridding interface. In contrast, the stability of many existing schemes with material traverse can only be verified numerically~\cite{ye2016temporal,kim2012subgridding3D} and some schemes exhibit inaccuracy when objects intersect the interface~\cite{wang2010analysis}.

\section{Numerical Examples}
\label{sec:examples}

The subgridding algorithm from Sec.~\ref{sec:subgridding} was implemented in MATLAB in order to test the validity of the proposed theory for FDTD stability and to assess the accuracy of the proposed subgridding method. Time-consuming portions of the code, such as the update equations for the fields inside each grid, were written using vectorized operations.

\subsection{Stability Verification}
\label{subsec:ex_cavity}

We verify the stability of the proposed algorithm by simulating the subgridding scenario shown in Fig.~\ref{fig:cavity_setup} for a million time steps. The setup consists of a 12~cm~$\!\times\!$~12~cm~$\!\times\!~$12~cm cavity with PEC walls. The cavity is discretized with a uniform coarse  mesh with $\dx\!=\!\dy\!=\!\dz \!=\!$ 1~cm. A 4~cm~$\!\times\!$~4~cm~$\!\times\!$~4~cm subgridding region with refinement ratio $r\!=\!5$ in all directions is placed at the center of the cavity. In order to test the correctness of stability analysis in the case when material properties vary, we randomly assign permittivity values to cells in both grids between the free space value $\eps_0$ and $3\eps_0$ using the $rand()$ function in MATLAB. In another simulation, we test if losses are properly handled by introducing, in addition to the varying permittivities, random conductivities between zero and 5$\times$10$^{-5}$~S/m. We excite the cavity using a Gaussian pulse with half-width at half-maximum bandwidth of 3.53~GHz. Time step is set to 99\% of the fine grid's CFL limit.

\begin{figure}[t!]
	
	\centering
	\includegraphics[scale=1]{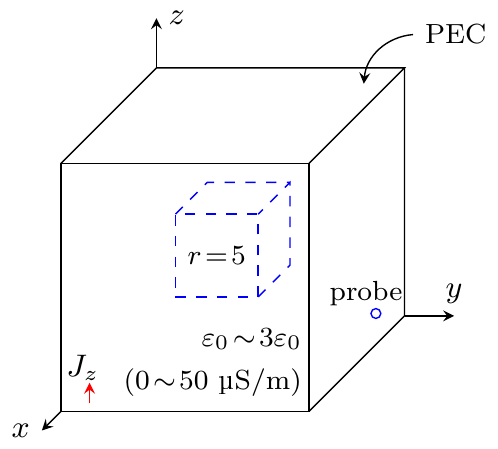}
	
	\caption{Setup of the stability verification test in Sec.~\ref{subsec:ex_cavity}.}
	\label{fig:cavity_setup}

\end{figure}

From the results, presented in Fig.~\ref{fig:results_cavity}, it can be seen that no signs of instability occur during the million time steps of the simulation, even if nonuniform materials traverse the coarse-fine interface. This result validates the correctness of the stability analysis framework presented in this paper. 

\begin{figure}[t!]
	
	\includegraphics[width = 0.95\columnwidth]{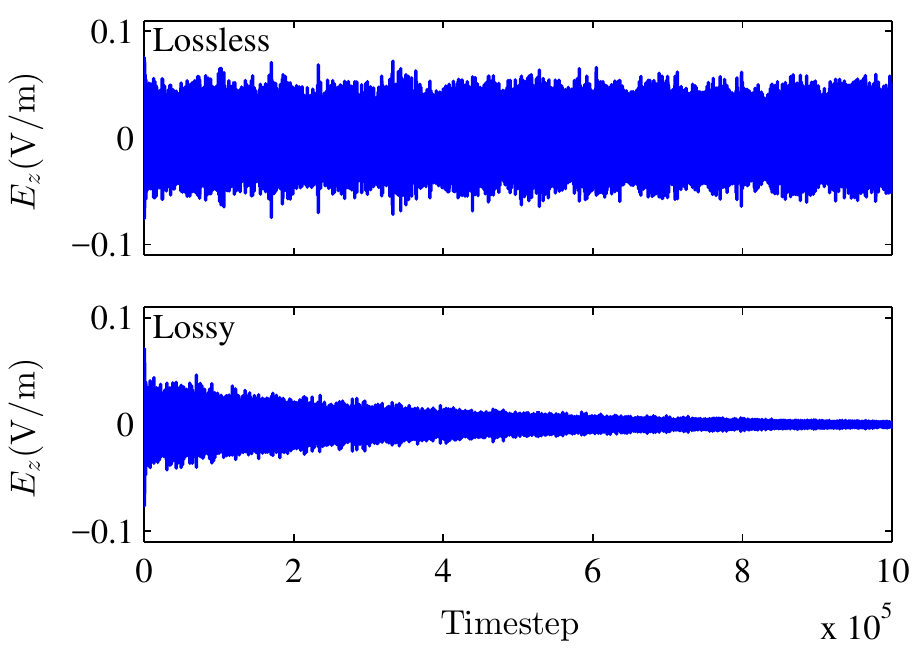}
	\caption{The $z$ component of electric field at the probe in the stability verification test of Sec.~\ref{subsec:ex_cavity}.}
	\label{fig:results_cavity}
\end{figure}

\subsection{Material Traverse}  
\label{subsec:ex_mt}

In this simulation, we test whether the proposed subgridding method produces consistent results when the subgridding interface traverses material boundaries, which is known to be a challenging scenario for subgridding schemes~\cite{wang2010analysis}. We use the setup shown in  Fig.~\ref{fig:setup_material_traverse}, where the subgridding region is placed in three different ways: in Case~1 and Case~3, a block of material traverses different sides of the interface. In the reference Case~2, the block is fully enclosed by the subgridding boundary. We perform the test for a copper block and for a lossy dielectric block ($\eps_r\!=\!3$, $\sigma \!=\! 0.05$~S/m). A 15-cm perfectly matched layer (PML) terminates the domain. The coarse grid is set to $\dx\!=\!\dy\!=\!\dz\!=\!$~1~cm and the subgridding region is refined by a factor of $r\!=\!3$ in each direction. A Gaussian source with half-width at half-maximum bandwidth of 1.02~GHz is used. The iteration time step is 99\% of the CFL limit of the fine grid.

\begin{figure}[t]
	\centering
	\includegraphics[]{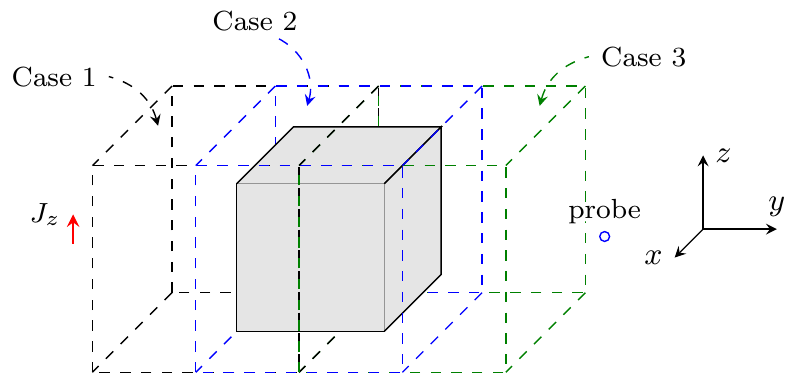}
	\caption{Setup of the material traverse test of Sec.~\ref{subsec:ex_mt}.}
	\label{fig:setup_material_traverse}
\end{figure}

Fig.~\ref{fig:results_material_traverse} shows the electric field recorded at the probe in the three cases. The two cases with material traverse are in very good agreement with the reference simulation, showing that the proposed algorithm does not lose accuracy when different materials traverse the subgridding boundary. The test is successful for both lossy dielectric and for copper.

	\begin{figure}[t!]
		
		\includegraphics[width=0.95\columnwidth]{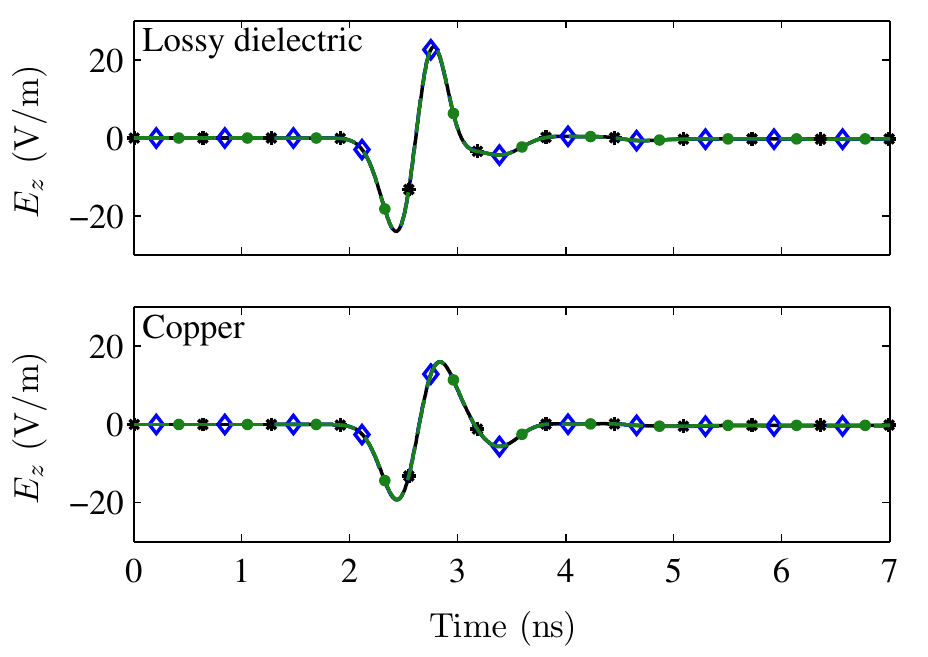}
		
		\caption{The $z$ component of electric field at the probe in material traverse test of Sec.~\ref{subsec:ex_mt} for Case 1 (\st{$\;\;\boldmath{\ast}\;\;$}), Case 2 ({\color{blue} \st{$\;$}$\;\diamond\;$\st {$\;$}}), and Case 3 ({\color{legendGreen} \st{$\;$}$\;\bullet\;$\st {$\;$}}).}
		\label{fig:results_material_traverse}
	\end{figure}

\subsection{Meta-Screen}
\label{subsec:ex_metascreen}

In this example, we use subgridding to study the three-slot meta-screen shown in Fig.~\ref{fig:setup_metascreen}. The setup for this simulation is similar to~\cite{Ludwig}, except for differences in discretization and screen material. In particular, we set the plate conductivity to 1.3$\times$10${\rm^6}$ S/m and thickness of the plate to 6 mils, based on the value in~\cite{Markley}. The slot in the center is 13.2~mm by 1.2~mm and the two satellite slots on the sides are 17~mm by 0.6~mm, placed at a center-to-center distance of 3~mm away from the central slot. These dimensions were designed to achieve subwavelength focusing at 10~GHz~\cite{Markley}. We place a line of probes 4.572~mm away from the meta-screen, which corresponds to approximately 15\% of the wavelength at 10~GHz -- distance at which measurements were performed in~\cite{Markley}. The simulation region was 25.2~mm$\times$19.2024~mm$\times$43.4~mm including PMLs.

As in~\cite{Ludwig}, we terminate the region before the meta-screen with PEC walls on the East and the West sides, and with perfect magnetic conductor walls on the Top and the Bottom sides. A five-cell PML layer is added on all sides in order to mimic an unbounded domain. 

The satellite slots have widths of only 2\% of the wavelength, which makes the problem intrinsically multiscale. Moreover,~\cite{Ludwig} reports resonance effects that make it difficult to resolve the structure near the design frequency. We choose this example to investigate the performance of the proposed subgridding scheme.

We set the coarse grid to $\dx \!=\! \dz \!=\!$~0.7~mm, and $\dy \!=\!$~0.1524~mm. In order to properly resolve the meta-screen structure and the fields, we refine a 9.8~mm$\times$2.5908~mm$\times$21~mm area around the slots with $r_x \!=\! r_z \!=\! 7$ and $r_y \!= 3$, using the proposed subgridding method. As a reference, we perform a conventional FDTD simulation with a nonuniform grid, where a region of dimensions 9.8~mm$\times$1.0668~mm$\times$21~mm is refined by $r_x \!=\! r_z \!=\! 7$ and $r_y \!= 3$ by means of varying $\dx$, $\dy$, and $\dz$. We also simulate the case where the entire structure is discretized with the coarse mesh. 

A modulated Gaussian waveform with 10~GHz central frequency and 8.24~GHz half-width at half-maximum bandwidth is used to excite the structure with a sheet of $x$-directed uniform current density on the South side. The proposed method and FDTD with nonuniform discretization are run at $\dt =$~0.1362~ps, whereas the coarse grid simulation has the iteration time step of 0.4810~ps. These time steps correspond to 99\% of the CFL limit of each mesh. The waveform is recorded for 5.5~ns, as in~\cite{Ludwig}.

\begin{figure}[t!]
	\centering
	\includegraphics[]{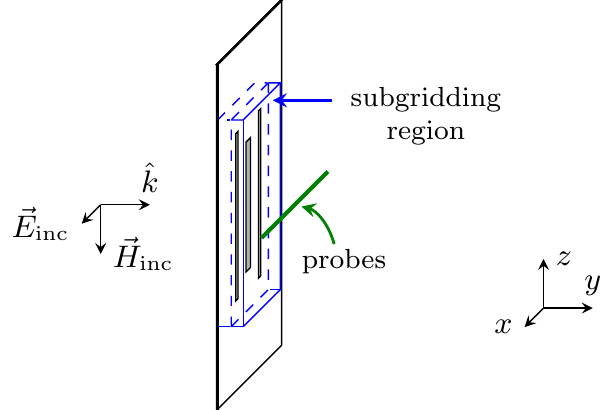}
	\caption{Setup of the meta-screen example of Sec.~\ref{subsec:ex_metascreen}.}
	\label{fig:setup_metascreen}
\end{figure}

Fig.~\ref{fig:results_metascreen} shows the amplitude of the $x$ component of electric field at 10~GHz, which was obtained with a Fourier transform of the time-domain waveforms. The results from subgridding are in very good agreement with the reference simulation. In contrast, the coarse resolution was insufficient to perform an accurate simulation of the meta-screen. This is not surprising, since any mesh with $\dx$ larger than 0.3~mm or $\dz$ larger than 0.1~mm cannot correctly resolve the meta-screen structure, let alone the fields around the slots.

Simulation times, as well as the number of primary FDTD cells are shown in Table~\ref{tab:performance_metascreen}. Subgridding almost doubles the speed-up compared to the nonuniform grid refinement, which is already much faster than the simulation with uniform grid. This is a result of the lower number of unknowns associated with subgridding. The guarantee of stability of the proposed method allows for the reduction of simulation time without sacrificing reliability. Moreover, the ability of the method to handle material traverse allows placing the subgridding region across the meta-screen.

\begin{figure}[t!]
	\centering
	{\includegraphics[width = 0.9\columnwidth]{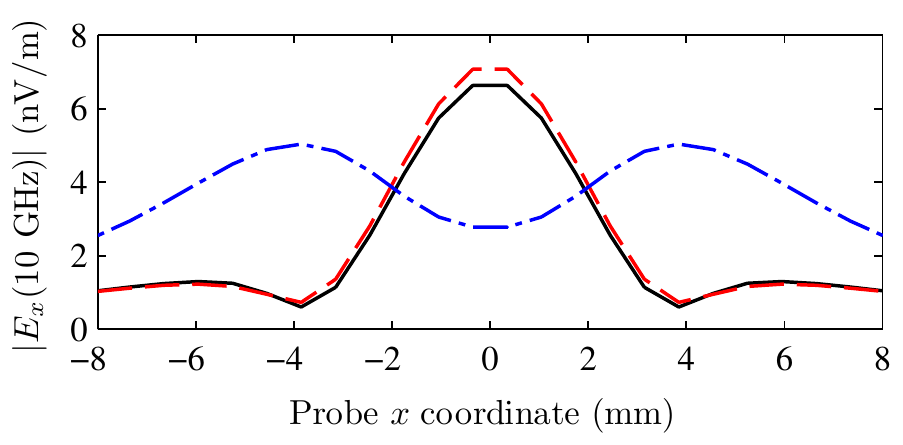}}
	\caption{Magnitude of $x$-directed electric field along the line of probes in the focusing meta-screen example of Sec.~\ref{subsec:ex_metascreen} for uniform coarse grid ({\color{blue} \st{$\;\;\;$}$\;$\st{$\;$}}), nonuniform grid (\st{$\;\;\;\;\;$}), and the proposed method ({\color{red} \st{$\;\;$}$\;$\st{$\;\;$}}).}
	\label{fig:results_metascreen}
\end{figure}

\begin{table}[t!]
	\centering
	\color{black}
	\caption{Computational Cost for Different Discretization in the Meta-Screen Example of Sec.~\ref{subsec:ex_metascreen}.}
\begin{threeparttable}
	\begin{tabular}{c c c c}
		\hline 
		Method 		& Number of primary cells 	& Simulation time & Speed-up\\
		\hline \hline
		Coarse		& 	\phantom{00,}281,232	&   \phantom{0}0.18 hours & 524.8\\ 
		\hline
		Fine 		& 	41,341,104		&  95.81 hours\tnote{a}& reference\\
		\hline
		Nonuniform	& 	\phantom{0}4,065,600 	& 11.46 hours	& 8.4\\
		\hline
		Proposed	&	\phantom{0}1,323,672	& \phantom{0}5.80 hours	&  16.5\\
		\hline
	\end{tabular}
\begin{tablenotes}
	\item[a] Estimated runtime from 30 time steps.
\end{tablenotes}
\end{threeparttable}
\color{yellow}
	\label{tab:performance_metascreen}
\end{table}

\section{Conclusion}
We proposed a dissipation theory for 3-D FDTD inspired by the theory of dissipative dynamical systems~\cite{willems1972dissipative}. The theory shows that the FDTD equations for the fields in a 3-D region can be interpreted as a dynamical system that is dissipative when time step satisfies the CFL limit.

The proposed dissipation theory provides a powerful and systematic method to create new FDTD schemes with guaranteed stability. By showing that each component of the scheme is dissipative (such as grids, boundary conditions, or embedded lumped components), the overall coupled scheme is guaranteed to be dissipative, and thus stable. This approach makes stability analysis simpler, more intuitive, and modular, since dissipation conditions are imposed on each component independently and are based on the familiar concept of energy.

With the new framework, we develop a stable 3-D subgridding scheme supporting material traverse. Numerical examples show good performance of the scheme when accelerating multiscale problems, as well as its ability to produce accurate and stable simulations when various materials traverse the subgridding interface. The dissipation framework is not restricted to subgridding algorithms, and could be applied to developing more advanced schemes in the future. As a proof of concept, we have applied the dissipation framework in 2-D to guarantee the stability of FDTD coupled to reduced-order models~\cite{jnl-2016-tap-fdtdmorstable}.

\appendix

\color{black}

\label{appendix:edges}

\newcommand{\phprime}{\phantom{'}}

In this section, we handle the case when the fine grid is surrounded by the coarse grid on all six sides. In this scenario, the locations where two faces of the boundary come together require a special treatment. Fig.~\ref{fig:edge_interp} shows the portion of the fine grid's boundary on its South-West edge, which is discussed in detail in this section.

\subsection{Interpolation Conditions}
\label{subsec:edge_interp}

Unlike in Sec.~\ref{sec:subgridding}, the edge case involves two sets of hanging variables: those in the $y$ direction on the West side and those in the $x$ direction on the South side. As in Sec.~\ref{sec:subgridding}, we interpolate the fields on the shaded portion of the interface so that the $z$-directed electric field samples are equal in the $z$ direction
\begin{subequations}
	\begin{eqnarray}
	\Ehat_{z1\phprime}^n = \Ehat_{z3\phprime}^n = \Ehat_{z5\phprime}^n \,, &\quad \forall n\,, \label{eq:edge_E1_eq} \\
	\Ehat_{z2\phprime}^n = \Ehat_{z4\phprime}^n = \Ehat_{z6\phprime}^n \,, &\quad \forall n\,, \label{eq:edge_E2_eq} \\
	\Ehat_{z2'}^n = \Ehat_{z4'}^n = \Ehat_{z6'}^n   \,,& \quad \forall n\,.\label{eq:edge_E3_eq}
	\end{eqnarray}
	\label{eq:edge_E_eq}
\end{subequations}
The fine $y$-directed hanging variables are set equal in the $y$ direction and the fine $x$-directed hanging variables are set equal in the $x$ direction
\begin{subequations}
	\begin{eqnarray}
	\Uhat_{y1\phprime}^{n+\half} = \Uhat_{y2\phprime}^{n+\half}  \,, \quad &
	\Uhat_{x1\phprime}^{n+\half} = \Uhat_{x2'}^{n+\half} \,, \quad &
	\forall n\,,  \label{eq:edge_U1_eq}\\
	\Uhat_{y3\phprime}^{n+\half} = \Uhat_{y4\phprime}^{n+\half}  \,, \quad &
	\Uhat_{x3\phprime}^{n+\half} = \Uhat_{x4'}^{n+\half} \,, \quad & 
	\forall n\,, \label{eq:edge_U2_eq} \\
	\Uhat_{y5\phprime}^{n+\half} = \Uhat_{y6\phprime}^{n+\half}  \,, \quad &
	\Uhat_{x5\phprime}^{n+\half} = \Uhat_{x6'}^{n+\half} \,, \quad &
	\forall n\,.   \label{eq:edge_U3_eq} 
	\end{eqnarray}
\end{subequations}      
We collect the distinct electric field samples into vectors as follows
\begin{equation}
	\Ehat_{z(SW)}^n = \hat{E}_{z1}^n\,,
	\quad
	\Evhat_{z(W)}^n =
	\hat{E}_{z2}^n\,,
	\quad
	\Evhat_{z(S)}^n =
	\hat{E}_{z2'}^n\,,
	\label{eq:def_edge_E}
\end{equation}
where vector $\Evhat_{z(W)}^n$ contains the distinct samples on the West half of the shaded region in Fig.~\ref{fig:edge_interp}, except for the South-West sample at ${(1,1,\khat\!+\!\halftxt)}$ and similarly for $\Evhat_{z(S)}^n$ on the South side. Variable $\Ehat_{z(SW)}^n$ is always the single sample at ${(1,1,\khat\!+\!\halftxt)}$, which is excluded from  $\Evhat_{z(W)}^n$ and $\Evhat_{z(S)}^n$. In general, $\Evhat_{z(W)}^n$ and $\Evhat_{z(S)}^n$ are vectors of size ${(r_y\!-\!1)/2}$ and ${(r_x\!-\!1)/2}$, respectively, although in the $r_x\! =\! r_y \!= 3$ example they reduce to scalar quantities.

Similarly, the distinct hanging variables in~\eqref{eq:edge_U1_eq}--\eqref{eq:edge_U3_eq} are collected into vectors $\Uvhat_y^{n+\halftxt}$ and $\Uvhat_x^{n+\halftxt}$
	\begin{equation}
	\Uvhat_y^{n+\half} = \begin{bmatrix}
	\Uhat_{y1}^{n+\half}&
	\Uhat_{y3}^{n+\half}&
	\Uhat_{y5}^{n+\half}
	\end{bmatrix}^T\,,
	\label{eq:def_edge_Uy}
	\end{equation}
	\begin{equation}
	\Uvhat_x^{n+\half} = \begin{bmatrix}
	\Uhat_{x1}^{n+\half}&
	\Uhat_{x3}^{n+\half}&
	\Uhat_{x5}^{n+\half}
	\end{bmatrix}^T\,.
	\label{eq:def_edge_Ux}
	\end{equation}

Analogously to~\eqref{eq:face_aver_E}, we force the coarse electric field sample $E_z^n$ to be equal to the average of $\Ehat_z^n$ over the West portion of the shaded area in Fig.~\ref{fig:edge_interp}
	\begin{equation}
	\frac{\dy}{2} E_z^n = \frac{\dy}{2r_y} \Ehat_{z(SW)}^n + \frac{\dy}{r_y} \T_{W}^T \Evhat_{z(W)}^n\,,\quad \forall n\,,
	\label{eq:edge_aver_E_W}
	\end{equation}
	with the weighting coefficients chosen to correctly account for the area of boundary faces associated with each $\Ehat_{z}^n$ sample. The column vector of ones $\T_{W}$ has size ${(r_y-1)/2}$. A similar condition is applied on the South portion of the shaded area in Fig.~\ref{fig:edge_interp}
	\begin{equation}
	\frac{\dx}{2}  E_z^n = \frac{\dx}{2 r_x} \Ehat_{z(SW)}^n + \frac{\dx}{r_x} \T_{S}^T \Evhat_{z(S)}^n\,,\quad \forall n\,.
	\label{eq:edge_aver_E_S}
	\end{equation}

As in the planar case, the coarse hanging variables are forced to be equal to the average of the corresponding fine hanging variables
	\begin{eqnarray}
	U_y^{n+\half} &=& \frac{1}{r_z}\T_{r_z}^T \Uvhat_y^{n+\half}\,,\quad \forall n\,,\label{eq:edge_aver_Uy} \\
	U_x^{n+\half} &=& \frac{1}{r_z}\T_{r_z}^T \Uvhat_x^{n+\half}\,, \quad \forall n\,.\label{eq:edge_aver_Ux}	
	\end{eqnarray}

\begin{figure}

	\centering
	
	\begin{tabular}{rc}
		\hspace*{-1cm} \includegraphics[scale=0.95]{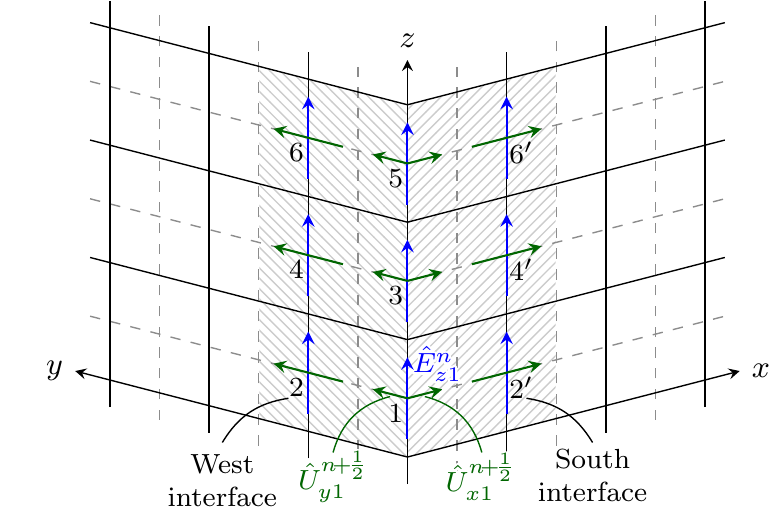}
	\end{tabular}

	\caption{Fine side of the subgridding interface in the Appendix.}
	\label{fig:edge_interp}
\end{figure}

\subsection{Stability Proof}
In this section, we show that with the interpolation rule in Sec.~\ref{subsec:edge_interp} for the edges of the boundary, the interpolation rule remains a lossless subsystem. Consider the portion of the coarse-fine interface consisting of the South shaded area in Fig.~\ref{fig:edge_interp} on the fine side and of the adjacent $\dx/2 \times \dz$ boundary face on the coarse side. For this portion of the interface, the net contribution of the $x$-directed hanging variables to the supply rate from the fine and coarse grids to the interpolation rule boils down to
\begin{multline}
s_{U_x}^{n+\half} = + \dt \dxhalf \dz \frac{E_z^n + E_z^{n+1}}{2} U_x^{n+\half} \\
- \dt \frac{\dx}{r_x} \frac{\dz}{r_z} \T^T_{r_z} \Uvhat_x^{n+\half} \frac{\left(\Evhat_{z(S)}^n + \Evhat_{z(S)}^{n+1}\right)}{2}^{\!T} \!\! \T_{S} \\
- \dt \frac{\dx}{2r_x} \frac{\dz}{r_z} \T_{r_z}^T \Uvhat_x^{n+\half}\frac{\Ehat_{z(SW)}^n+\Ehat_{z(SW)}^{n+1}}{2}  \,,
\label{eq:supply_edge}
\end{multline}
analogously to~\eqref{eq:supply_face}. Substituting~\eqref{eq:edge_aver_Ux}, \eqref{eq:supply_edge} becomes
\begin{multline}
s_{U_x}^{n+\half} = + \dt \dxhalf \dz \frac{E_z^n + E_z^{n+1}}{2} U_x^{n+\half} 
-\dt \frac{\dz}{r_z} r_z U_x^{n+\half}\\
\times \Bigg(
\frac{\dx}{r_x}  \frac{\left(\Evhat_{z(S)}^n + \Evhat_{z(S)}^{n+1}\right)}{2}^{\!T} \!\!\T_{S} 
+ \frac{\dx}{2r_x}  \frac{\Ehat_{z(SW)}^n+\Ehat_{z(SW)}^{n+1}}{2} \Bigg)\,,
\label{eq:supply_edge2}
\end{multline}
which is zero because of~\eqref{eq:edge_aver_E_S}. The proof for other hanging variables near the edges of the boundary can be done with a similar argument. This shows that even in presence of corners, the interpolation rule is a lossless system, since its supply rate is zero.

\subsection{Practical Implementation}
The update equations for the electric fields in the shaded area in Fig.~\ref{fig:edge_interp} are derived using the modified FDTD equations on each side of the interface and the interpolation conditions~\eqref{eq:edge_E1_eq}--\eqref{eq:edge_E3_eq}, \eqref{eq:edge_U1_eq}--\eqref{eq:edge_U3_eq}, \eqref{eq:edge_aver_E_W}, \eqref{eq:edge_aver_E_S}, \eqref{eq:edge_aver_Uy}, and \eqref{eq:edge_aver_Ux}. The coarse $E_z^n$ field at ${(i,j,k\!+\!\halftxt)}$ needs to satisfy the modified FDTD equation
\begin{multline}
3\dxhalf \dyhalf  \left( \frac{\eps_z}{\dt} \!+\! \frac{\sigma_z}{2} \right) E_z^{n+1} \\
= 3\dxhalf \dyhalf  \left( \frac{\eps_z}{\dt} \!-\! \frac{\sigma_z}{2} \right) E_z^{n}
+ \dyhalf U_y^{n+\half}\\
- \dxhalf H_x\big|_{j+\half}^{n+\half} 
- \dy H_y\big|_{i-\half}^{n+\half} 
+ \dx H_x\big|_{j-\half}^{n+\half} \\
+ \dyhalf H_y\big|_{i+\half}^{n+\half}
- \dxhalf U_x^{n+\half} \,.
\label{eq:edge_crs}
\end{multline}
Applying a similar averaging procedure as in Sec.~\ref{subsec:subgridding_update}, we obtain the following relations for the fine electric field vectors defined in~\eqref{eq:def_edge_E}
\begin{subequations}
	\begin{multline}
	\frac{\dx}{2 r_x} \frac{\dy}{2r_y} \left( \frac{\hat{\eps}_{z(SW)}}{\dt} \!+\! \frac{\hat{\sigma}_{z(SW)}}{2} \right) \Ehat_{z(SW)}^{n+1}\\
	= \frac{\dx}{2 r_x} \frac{\dy}{2r_y} \left( \frac{\hat{\eps}_{z(SW)}}{\dt} \!-\! \frac{\hat{\sigma}_{z(SW)}}{2} \right) \Ehat_{z(SW)}^{n}
	+ \frac{\dx}{2 r_x} U_x^{n+\half} \\
	+ \frac{\dy}{2r_y} \hat{H}_{y(SW),\ihat+\half}^{n+\half}
	- \frac{\dx}{2 r_x} \hat{H}_{x(SW),\jhat+\half}^{n+\half} 
	- \frac{\dy}{2r_y} U_y^{n+\half}
	\,,
	\label{eq:edge_SW}
	\end{multline}
	\begin{multline}
	\frac{\dx}{2 r_x}\frac{\dy}{r_y} \left(\frac{\Dhat_{\eps_{z(W)}}}{\dt} \!+\! \frac{\Dhat_{\sigma_{z(W)}}}{2} \right) \Evhat_{z(W)}^{n+1}\\
	= \frac{\dx}{2 r_x}\frac{\dy}{r_y} \left(\frac{\Dhat_{\eps_{z(W)}}}{\dt} \!-\! \frac{\Dhat_{\sigma_{z(W)}}}{2}  \right) \Evhat_{z(W)}^{n}
	+\frac{\dx}{2 r_x} \Hvhat_{x(W),\jhat-\half}^{n+\half}\\
	+\frac{\dy}{r_y} \Hvhat_{y(W),\ihat +\half}^{n+\half}
	-\frac{\dx}{2 r_x} \Hvhat_{x(W),\jhat+\half}^{n+\half}
	-\frac{\dy}{r_y} \T_{W} U_y^{n+\half}\,,
	\label{eq:edge_W}
	\end{multline}
	\begin{multline}
	\frac{\dx}{r_x}\frac{\dy}{2r_y} \left(\frac{\Dhat_{\eps_{z(S)}}}{\dt} \!+\! \frac{\Dhat_{\sigma_{z(S)}}}{2} \right) \Evhat_{z(S)}^{n+1}\\
	= \frac{\dx}{r_x}\frac{\dy}{2r_y} \left(\frac{\Dhat_{\eps_{z(S)}}}{\dt} 
	\!-\! \frac{\Dhat_{\sigma_{z(S)}}}{2}  \right) \Evhat_{z(S)}^{n}
	+\frac{\dy}{2r_y} \Hvhat_{y(S),\ihat+\half}^{n+\half}\\
	-\frac{\dx}{r_x} \Hvhat_{x(S),\jhat+\half}^{n+\half}
	-\frac{\dy}{2r_y} \Hvhat_{y(S),\ihat-\half}^{n+\half} 
	+\frac{\dx}{r_x} \T_{S}U_x^{n+\half}\,,
	\label{eq:edge_S}
	\end{multline}
\end{subequations}
where
\begin{subequations}
	\begin{eqnarray}
	\hat{H}_{x(SW),\jhat+\half}^{n+\half} &=& \frac{1}{r_z}\sum_{\phantom {'}\mhat = 1,3,5\phantom {''}} \!\!\!\!\hat{H}_{x\mhat}\big|_{\jhat+\half}^{n+\half}\,,
	\label{eq:Hdef_averaged_edge_first}
	\\
	\Hvhat_{x(W),\jhat-\half}^{n+\half} &=&
	\frac{1}{r_z}\sum_{\phantom {'}\mhat=2,4,6\phantom {''}}
	\!\!\!\!\hat{H}_{x\mhat}\big|_{\jhat-\half}^{n+\half}\,,
	\\
	\Hvhat_{x(W),\jhat+\half}^{n+\half} &=&
	\frac{1}{r_z}\sum_{\phantom {'}\mhat=2,4,6\phantom {''}}
	\!\!\!\!\hat{H}_{x\mhat}\big|_{\jhat+\half}^{n+\half}\,,
	\\
	\Hvhat_{x(S),\jhat+\half}^{n+\half} &=&
	\frac{1}{r_z}\sum_{\mhat=2',4',6'}
	\!\!\!\!\hat{H}_{x\mhat}\big|_{\jhat+\half}^{n+\half}\,,
	\label{eq:Hdef_averaged_edge_last}
	\end{eqnarray}
\end{subequations}
where $(\ihat\!=\!1,\jhat\!=\!1,\khat\!+\!\halftxt)$ is the coordinate of node~$\mhat$. In general, vectors $\Hvhat_{x(W),\jhat\pm\halftxt}^{n+\half}$ and $\Hvhat_{y(W),\ihat+\halftxt}^{n+\halftxt}$ have size ${(r_y\!-\!1)/2}$, with each element corresponding to a sample in $\Evhat_{z(W)}$. Quantities $\hat{H}_{y(SW),\ihat+\halftxt}^{n+\halftxt}$, $\Hvhat_{y(S),\ihat\pm\halftxt}^{n+\halftxt}$, and $\Hvhat_{y(W),\ihat+\halftxt}^{n+\halftxt}$ are defined similarly to~\eqref{eq:Hdef_averaged_edge_first}--\eqref{eq:Hdef_averaged_edge_last}. Permittivity and conductivity coefficients $\hat{\eps}_{z(SW)}$, $\hat{\sigma}_{z(SW)}$, $\Dhat_{\eps_{z(W)}}$, $\Dhat_{\sigma_{z(W)}}$, $\Dhat_{\eps_{z(S)}}$, and $\Dhat_{\sigma_{z(S)}}$ are obtained from averaging the values in the half-cells adjacent to the boundary over $\khat$ indexes. Equations~\eqref{eq:edge_SW}--\eqref{eq:edge_S} serve a purpose similar to~\eqref{eq:face_fine_mtx}.

We now combine~\eqref{eq:edge_crs} for the $z$-directed electric field on the coarse side, the recurrence relations \eqref{eq:edge_SW}--\eqref{eq:edge_S} for the fields on the fine side, and the interpolation rules~\eqref{eq:edge_aver_E_W} and \eqref{eq:edge_aver_E_S}, into a linear system of the form
\begin{equation}
\matr{A} \vect{z}^{n+1} = \vect{b}^n\,,
\label{eq:edge_syst}
\end{equation}
with an ${(r_x/2+r_y/2+3)\times1}$ vector of unknowns
\begin{equation}
\vect{z}^{n+1} = 
\begin{bmatrix}
E_z^{n+1}\\[2pt]
\Ehat_{z(SW)}^{n+1}\\[4pt]
\Evhat_{z(W)}^{n+1}\\[4pt]
\Evhat_{z(S)}^{n+1}\\[4pt]
U_x^{n+\half}\\[0pt]
U_y^{n+\half}\\[0pt]
\end{bmatrix}\,.
\end{equation}
Solving~\eqref{eq:edge_syst} for $\vect{z}^{n+1}$ gives values of $E_z^{n+1}$, $\Ehat_{z(SW)}^{n+1}$, $\Evhat_{z(W)}^{n+1}$, and $\Evhat_{z(S)}^{n+1}$.  With~\eqref{eq:edge_E1_eq}--\eqref{eq:edge_E3_eq}, samples $\Ehat_{z3}^{n+1}$ and $\Ehat_{z5}^{n+1}$ are obtained from $\Ehat_{z1}^{n+1}$, samples $\Ehat_{z4}^{n+1}$ and $\Ehat_{z6}^{n+1}$ from $\Ehat_{z2}^{n+1}$, and $\Ehat_{z4'}^{n+1}$ and $\Ehat_{z6'}^{n+1}$ from $\Ehat_{z2'}^{n+1}$. The size of system~\eqref{eq:edge_syst} is only $(r_x/2+r_y/2+3)\times (r_x/2+r_y/2+3)$ and the resulting overhead is typically small. 

\bibliographystyle{myIEEEtran}
\bibliography{IEEEabrv,bibliography}

\end{document}